\documentclass[preprint,11pt]{elsarticle}

\usepackage[title]{appendix}

\usepackage{amsfonts,amsthm,amssymb}
\usepackage{lineno,hyperref}

 \usepackage{amsthm}

\journal{Stochastic Processes and their Applications}
\newtheorem{theorem}{Theorem}[section]

\newtheorem{proposition}[theorem]{Proposition}
\newtheorem{remark}[theorem]{Remark}

\newtheorem{lemma}[theorem]{Lemma}

\newtheorem{definition}[theorem]{Definitions}
\newtheorem{corollary}[theorem]{Corollary}

\newcommand{\E}{\mathbb E}

\def \Rbrack {[\![}
\def \Lbrack {]\!]}

\def \Rbrack {[\![}
\def \Lbrack {]\!]}

\makeatletter
\newcommand*\bigcdot{\mathpalette\bigcdot@{.5}}
\newcommand*\bigcdot@[2]{\mathbin{\vcenter{\hbox{\scalebox{#2}{$\m@th#1\bullet$}}}}}
\makeatother
\newcommand{\is}{\bigcdot }

\begin{document}

\begin{frontmatter}



\title{Explicit description of all deflators for market models under random horizon with applications to NFLVR}




\author[UofA]{{Tahir Choulli}\corref{mycorrespondingauthor}}
\ead{tchoulli@ualberta.ca}

\author[UofA]{Sina Yansori}
\address[UofA]{Mathematical and Statistical Sciences Dept.\\
University of Alberta, Edmonton, Canada}

\cortext[mycorrespondingauthor]{Corresponding author}

\begin{abstract} This paper considers an initial market model, specified by its underlying assets $S$ and its flow of information $\mathbb F$, and an arbitrary random time $\tau$ which might not be an $\mathbb F$-stopping time. As the death time and the default time (that $\tau$ might represent) can be seen when they occur only, the progressive enlargement of $\mathbb F$ with $\tau$ sounds tailor-fit for modelling the new flow of information $\mathbb G$ that incorporates both $\mathbb F$ and $\tau$. In this setting of informational market, the first principal goal resides in describing as explicitly as possible the set of all deflators for $(S^{\tau}, \mathbb G)$, while the second principal goal lies in addressing the No-Free-Lunch-with-Vanishing-Risk concept (NFLVR hereafter) for $(S^{\tau}, \mathbb G)$.  Besides this direct application to NFLVR,  the set of all deflators constitutes the dual set of all ``admissible" wealth processes for the stopped model $(S^{\tau},\mathbb G)$, and hence it is vital in many hedging and pricing related optimization problems. Thanks to the results of  Choulli et al. \cite{ChoulliDavelooseVanmaele},  on martingales classification and representation for progressive enlarged filtration, our two main goals are fully achieved in different versions, when the survival probability never vanishes. The results are illustrated on the two particular cases when $(S,\mathbb F)$ follows the jump-diffusion model and  the discrete-time model. 
\end{abstract}

\end{frontmatter}

\section{Introduction} 
This paper considers an initial market model  represented by the pair $(S,\mathbb F)$, where $S$ represents the discounted stock prices for $d$-stocks, and $\mathbb F$ is the flow of ``public" information which is available to all agents. To this initial market model, we add a random time $\tau$ (i.e., a random variable with values in $[0,+\infty]$) that is revealed only at the terminal date. Mathematically speaking, this means that $\tau$ might not be an $\mathbb F$-stopping time. Thus, for modelling the new flow of information, we adopt the progressive enlargement of  $\mathbb F$ with $\tau$, that we denote throughout the paper by $\mathbb G$. Hence, our resulting informational market model is the pair $(S^{\tau},\mathbb G)$ \footnote{Throughout the paper, for any random time $\sigma$ and any process $X$, we denote by $X^{\sigma}$ the stopped process given by $X^{\sigma}_t:=X_{\sigma\wedge t},\ t\geq 0$.}. This information modelling allows us to apply our obtained results  to credit risk theory and life insurance (mortality and/or longevity risk), where the progressive enlargement of filtration sounds tailor-fit, and the initial enlargement of filtration --as in the insider trading framework-- is totally inadequate. In fact the death time of an agent can not be seen with certainty before its occurrence, and there is no single financial literature that models the information in the default of a firm $\tau$ as fully seen from the beginning as in the case of insider trading. For the case of initial enlargement and its application in insider trading problems, we refer the reader to \cite{Amendinger, Ankirchner} and the references therein to cite a few.\\

 For this new market model $(S^{\tau},\mathbb G)$, many challenging questions arise in finance (both theoretical and empirical) and mathematical finance. Most of these questions are still open problems nowadays and are essentially concerned with measuring the impact of $\tau$ on the financial and economical  concepts, theories, rules, models, methodologies, ...., et cetera.  Among these we cite the (consumption-based) capital asset pricing model (s), equilibrium, arbitrage theory, market's viability, the fundamental theorem of asset pricing, the optimal portfolios (e.g. the log-optimal portfolio, the num\'eraire portfolio and other types of portfolios to cite a few), utility maximization, the various pricing rules, ..., et cetera. The first fundamental question to all these aforementioned problems, lies in the impact of the random time on the market's viability and the corresponding no-arbitrage concept(s). In virtue of \cite{ChoulliDengMa}, see also \cite{KabanovKardarasSong,KardarasKaratzas} for related discussions, the market's viability in its various weakest form, the existence of the num\'eraire portfolio, and the no-unbounded-profit-with-bounded-risk (NUPBR for short) concept are equivalent and/or intimately related. In this spirit, there were an upsurge interest in studying first the effect of the random time on NUPBR in a series of papers, see \cite{AFK,ACDJ1,ACDJ3,ChoulliDeng2020,CD1,Song2016} for details. This very recent literature answers fully the question when NUPBR is altered for  $(S^{\tau},\mathbb G)$. Some of these papers, especially \cite{ACDJ1,ACDJ3,Song2016}, construct {\it examples of deflators} for special and very particular cases (such as when $(S,\mathbb F)$ is local martingale under the physical probability), while the following question remains open up to now. 
\begin{equation}\label{question1}
\mbox{How can we describe the set of all deflators for the model $(S^{\tau},\mathbb G)$?}\end{equation}
The importance of this set and its numerous roles in optimization problems intrinsic to financial problems sound clear and without reproach. Indeed, the set of deflators represents somehow the dual set of all ``admissible" wealth processes. No matter what is the optimization criterion, any optimal portfolio corresponds uniquely to an optimal deflator, and they are linked to each other via `` some duality form". Furthermore, in many (probably all) cases even when the utility is nice enough such as log utility,  it is more convenient, more efficient, and easier to solve a dual problem and describe the optimal deflator than getting the optimal portfolio directly. When considering the impact of  $\tau$ on optimal portfolio, we refer the reader to \cite{ChoulliSina2} for direct  application of the current paper. Besides these known applications, knowing the set of all deflators will allow us to address the NFLVR concept for the stopped model $(S^{\tau}, \mathbb G)$. Thus, under the assumption that the survival probability never vanishes, i.e., 
\begin{eqnarray}\label{Gpositive}
 P\left(\tau>t\ \big|\ {\cal F}_t\right)>0\quad P\mbox{-a.s.}\quad \mbox{for all}\quad t\geq 0,\end{eqnarray}
we answer (\ref{question1}) explicitly and completely. For this question, the assumption (\ref{Gpositive}) can be droped at some expenses, as it will be discussed in Section \ref{SectionOfDeflators} after Theorem \ref{GeneralDeflators}. On the one hand, under (\ref{Gpositive}), we prove that NFLVR holds for $(S^{T\wedge\tau}, \mathbb G)$ when it holds for $(S^T, \mathbb F)$, for any $T\in (0,+\infty)$.  This answers fully, under (\ref{Gpositive}) and for bounded investment's horizon, the question of how $\tau$ impacts NFLVR that has been around since a while. On the other hand, we give sufficient conditions on $\tau$ other than  (\ref{Gpositive}) such that $(S^{\tau}, \mathbb G)$ fulfills NFLVR as soon as $(S, \mathbb F)$ does. It is important to mention that in \cite{CJN}, see also \cite[Lemma 3.1]{ACDJ0}, the authors address this question for the immersion case (i.e., the case when every $\mathbb F$-martingale is a $\mathbb G$-martingale), while in \cite{Fontana} the authors proved that NFLVR always fails when $\mathbb F$ is a Brownian filtration and $\tau$ is an honest time avoiding $\mathbb F$-stopping times. In \cite[Lemma 3.3]{ACDJ0}, it is proved that  $(S, \mathbb G)$ satisfies NFLVR when $(\tau, \mathbb F,P)$ fulfills the  positive density hypothesis \footnote{i.e., there is a positive two-parameterized and measurable process $\alpha(t,s)$ such that $P(\tau>u|{\cal F}_t)=\int_u^{\infty} \alpha(t, s)dP(\tau\leq s),\ t\geq u\geq 0$ and $dP(\tau\leq t)=f(t)dt$} and $S$ is an $(\mathbb F, P)$-martingale. For more models for $(S, \mathbb F,\tau)$ for which the NFLVR stability holds or fails, we refer to \cite{ACDJ0}. \\

This paper contains four sections including the current one. Section  \ref{section2} presents the mathematical model and its preliminaries. Section \ref{SectionOfDeflators} has three subsections. The first subsection states the main results on the explicit parametrization of all deflators for $(S^{\tau},\mathbb G)$, while the second subsection proves these results. The last subsection illustrates these results on particular cases. Section \ref{Section4NFLVR} addresses the NFLVR for $(S^{\tau}, \mathbb G)$ and its variant. Furthermore, the paper contains an appendix where some proofs are relegated, and  some useful results (new and existing)  are detailed. 

\section{ The financial setting and preliminaries}\label{section2}
This section defines the notations, the financial and the mathematical concepts that the paper addresses or uses, the mathematical model that we focus on,  and some useful existing results. 
Throughout the paper, we consider the complete probability space  $\left(\Omega, {\cal F}, P\right)$. By  ${\mathbb H}$ we denote an arbitrary  filtration that satisfies the usual conditions of completeness and right continuity.  For any process $X$, the $\mathbb H$-optional projection and  the $\mathbb H$-predictable projection, when they exist, will be denoted by $^{o,\mathbb H}X$  and $^{p,\mathbb H}X$ respectively. The set ${\cal M}(\mathbb H, Q)$ (respectively  ${\cal M}^{(q)}(\mathbb H, Q)$ for $q\in (1,+\infty)$) denotes the set of all $\mathbb H$-martingales (respectively $q$-integrable martingales) under $Q$, while ${\cal A}(\mathbb H, Q)$ denotes the set of all $\mathbb H$-optional processes that are right-continuous with left-limits (RCLL for short) with integrable variation under $Q$. When $Q=P$, we simply omit the probability for the sake of simple notations.  For an $\mathbb H$-semimartingale $X$, by $L(X,\mathbb H)$ we denote the set of $\mathbb H$-predictable processes that are $X$-integrable in the semimartingale sense.  For $\varphi\in L(X,\mathbb H)$, the resulting integral of $\varphi$ with respect to $X$ is denoted by $\varphi\is X$. For $\mathbb H$-local martingale $M$, we denote by $L^1_{loc}(M,\mathbb H)$ the set of $\mathbb H$-predictable processes $\varphi$ that are $X$-integrable and the resulting integral $\varphi\is M$ is an $\mathbb H$-local martingale. If ${\cal C}(\mathbb H)$ is a set of processes that are adapted to $\mathbb H$, then ${\cal C}_{loc}(\mathbb H)$ is the set of processes, $X$, for which there exists a sequence of $\mathbb H$-stopping times, $(T_n)_{n\geq 1}$, that increases to infinity and $X^{T_n}$ belongs to ${\cal C}(\mathbb H)$, for each $n\geq 1$.  The $\mathbb H$-dual optional projection and the $\mathbb H$-dual predictable projection of a process $V$ with finite variation, when they exist, will be denoted by  $V^{o,\mathbb H}$  and $V^{p,\mathbb H}$ respectively. For any real-valued $\mathbb H$-semimartingale, $L$, we denote by ${\cal E}(L)$ the Dol\'eans-Dade (stochastic) exponential. It is the unique solution to the stochastic differential equation $dX=X_{-}dL,\ X_0=1,$  and is given by
\begin{eqnarray}\label{DDequation}
 {\cal E}_t(L)=\exp\left(L_t-L_0-{1\over{2}}\langle L^c\rangle_t\right)\prod_{0<s\leq t}(1+\Delta L_s)e^{-\Delta L_s}.\end{eqnarray}
Below, we recall the mathematical definition of deflator, where we distinguish the cases of local martingale deflators and deflators.
\begin{definition}\label{DeflatorDefinition} Consider the model $(X, \mathbb H, Q)$, where $\mathbb H$ is a filtration, $Q$ is a probability, and $X$ is a $(Q,\mathbb H)$-semimartingale. Let $Z$ be a process.\\
{\rm{(a)}}   We call $Z$ a local martingale deflator for $(X,Q,\mathbb H)$ if $Z>0$ and there exists a real-valued and $\mathbb H$-predictable process $\varphi$ such that $0<\varphi\leq 1$ and both $Z$  and $Z(\varphi\is X)$ are $\mathbb H$-local martingales under $Q$.  Throughout the paper, the set of these local martingale deflators will be denoted by ${\cal Z}_{loc}(X,Q,\mathbb H)$.\\
{\rm{(b)}}   We call $Z$ a deflator for $(X,Q,\mathbb H)$ if $Z>0$ and $Z{\cal E}(\varphi\is X)$ is an $\mathbb H$-supermartingale under $Q$, for any $\varphi\in L(X, \mathbb H)$ such that $\varphi\Delta X\geq -1$. The set of all deflators will be denoted by ${\cal D}(X,Q,\mathbb H)$. When $Q=P$, for the sake of simplicity, we simply omit the probability in notations and terminology.
\end{definition}
 The rest of this section has two subsections, where we present our mathematical model and recall a martingale representation result respectively. 
\subsection{The mathematical model}
Throughout the paper, on $\left(\Omega, {\cal F}, P\right)$ we consider $\mathbb F:=({\cal F}_t)_{t\geq 0}$ a filtration, satisfying the usual conditions of right continuity and completeness, which models  the public flow of information. Furthermore, we suppose given an $\mathbb F$-semimartingale, $S$, that represents the discounted price process of risky assets. In addition to this initial market model (i.e., $(S, \mathbb F)$), we consider a random time $\tau$, that might represent the death time of an agent or the default time of a firm, and hence it might not be an $\mathbb F$-stopping time. In the rest of the paper, we will use the pair $(D,\mathbb G)$ given by
\begin{equation}\label{processD}
D:=I_{\Rbrack\tau,+\infty\Rbrack},\ \mathbb G:=({\cal G}_t)_{t\geq 0},\ {\cal G}_t:={\cal G}^0_{t+}\ \mbox{with} \ {\cal G}_t^0:={\cal F}_t\vee\sigma\left(D_s,\ s\leq t\right).
\end{equation}
The agent endowed with $\mathbb F$, can only get information about $\tau$ via the survival probabilities $G$ and $\widetilde G$, called Az\'ema supermartingales, and given by
\begin{eqnarray}\label{GGtilde}
G_t :=^{o,\mathbb F}(I_{\Rbrack0,\tau\Rbrack})_t= P(\tau > t | {\cal F}_t) \ \mbox{ and } \ \widetilde{G}_t :=^{o,\mathbb F}(I_{\Rbrack0,\tau\Lbrack})_t= P(\tau \ge t | {\cal F}_t).\end{eqnarray}
The process
\begin{equation} \label{processm}
m := G + D^{o,\mathbb F},
\end{equation}
is a BMO $\mathbb F$-martingale, and thanks to  \cite[Theorem 3]{ACJ} and \cite[Theorem 2.3 and Theorem 2.11]{ChoulliDavelooseVanmaele}, we recall 
\begin{theorem}\label{Toperator} The following assertions hold.\\
{\rm{(a)}} For any $M\in{\cal M}_{loc}(\mathbb F)$,  the process
\begin{equation} \label{processMhat}
{\cal T}(M) := M^\tau -{\widetilde{G}}^{-1} I_{\Lbrack 0,\tau\Lbrack} \is [M,m] +  I_{\Lbrack 0,\tau\Lbrack} \is\Big(\sum \Delta M I_{\{\widetilde G=0<G_{-}\}}\Big)^{p,\mathbb F}\end{equation}
 is a $\mathbb G$-local martingale.\\
 {\rm{(b)}}  The process 
\begin{equation} \label{processNG}
N^{\mathbb G}:=D - \widetilde{G}^{-1} I_{\Lbrack 0,\tau\Lbrack} \is D^{o,\mathbb  F}
\end{equation}
is a $\mathbb G$-martingale with integrable variation. Moreover, $H\is N^{\mathbb G}$ is a $\mathbb G$-local martingale with locally integrable variation for any $H$ belonging to
\begin{equation} \label{SpaceLNG}
{\mathcal{I}}^o_{loc}(N^{\mathbb G},\mathbb G) := \Big\{K\in \mathcal{O}(\mathbb F)\ \ \big|\quad \vert{K}\vert G{\widetilde G}^{-1} I_{\{\widetilde{G}>0\}}\is D\in{\cal A}_{loc}(\mathbb G)\Big\}.
\end{equation}
\end{theorem}
The following lemma discusses some properties (new and old), about the parameters of $\tau$ (i.e., $G$ and $D^{o,\mathbb F}$), and is very useful throughout the paper.
 \begin{lemma}\label{Gdecomposition} The following assertions hold.\\
{\rm{(a)}}  $G$ is a positive process (i.e., (\ref{Gpositive}) holds) if and only if $ G_{-}$ is a positive process if and only if  $\widetilde G$ is a positive process.  Furthermore, we have 
\begin{eqnarray}\label{Rzero}
R_0=\inf\{t\geq 0: G_t=0\}=\inf\{t> 0: G_{t-}=0\}=\inf\{t\geq 0:{\widetilde G}_t=0\}.\end{eqnarray}
{\rm{(b)}}  If $G$ is positive, then we always have \footnote{This multiplicative decomposition differs from that of the literature, see for  instance \cite[Theorem 8.21, Chapter II]{JS03}, as the process ${\cal E}(-{\widetilde G}^{-1}\is D^{o,\mathbb F})$ is not predictable in general.}
\begin{eqnarray}\label{MultiDecompo4G}
G=G_0{\cal E}(G_{-}^{-1}\is m){\cal E}(-{\widetilde G}^{-1}\is D^{o,\mathbb F}).\end{eqnarray}
{\rm{(c)}}  Suppose $G>0$. The three processes  $G$, $ {\cal E}(G_{-}^{-1}\is m)$ and ${\cal E}(-{\widetilde G}^{-1}\is D^{o,\mathbb F})$  have their limits at infinity, $P$-almost surely,  and
 \begin{eqnarray}
G_{\infty}:=\lim_{t\rightarrow\infty}G_t=P(\tau=\infty\ \big| {\cal F}_{\infty-})= G_0{\cal E}_{\infty}({1\over{G_{-}}}\is m){\cal E}_{\infty}(-{1\over{\widetilde G}}\is D^{o,\mathbb F}), \label{Ginfinity}\end{eqnarray} 
where ${\cal E}_{\infty}(G_{-}^{-1}\is m):=\displaystyle\lim_{t\rightarrow\infty}{\cal E}_t(G_{-}^{-1}\is m)$ and ${\cal E}_{\infty}(-{\widetilde G}^{-1}\is D^{o,\mathbb F}):=\displaystyle\lim_{t\rightarrow\infty}{\cal E}_t(-{\widetilde G}^{-1}\is D^{o,\mathbb F}).$\\
{\rm{(d)}}  Suppose $G>0$. Then we have
\begin{eqnarray}\label{inclusion4G}
\{\tau=\infty\}\subset\{G_{\infty}>0\}=\{{\cal E}_{\infty}({1\over{G_{-}}}\is m)>0\}\cap\{{\cal E}_{\infty}(-{1\over{\widetilde G}}\is D^{o,\mathbb F})>0\}.\end{eqnarray}
{\rm{(e)}} Suppose $G>0$. There is equivalence between $\tau<\infty$ $P$-a.s., and 
\begin{eqnarray}\label{Exhausion}
\Omega=\biggl\{ {\cal E}_{\infty}(-{\widetilde G}^{-1}\is D^{o,\mathbb F})=0\biggr\}\bigcup\biggl\{  {\cal E}_{\infty}(G_{-}^{-1}\is m)=0\biggr\}.\end{eqnarray}
\end{lemma}
The poof of the lemma is relegated to Appendix \ref{AppB}. Some useful and immediate consequences of Lemma \ref{Gdecomposition}-(a), about the ${\cal T}$ransformation operator ${\cal T}$ \footnote{This operator can be defined, more generally, on the vector space of $\mathbb F$-semimartingales ${\cal S}(\mathbb F)$, and it satisfies ${\cal T}({\cal M}_{loc}(\mathbb F))\subset{\cal M}_{loc}(\mathbb G)$ as Theorem \ref{Toperator}-(a) claims. However, in general, ${\cal T}(M)$ is not the local martingale part in the $\mathbb G$-canonical decomposition of $M^{\tau}$.} defined in (\ref{processMhat}), are given in the following remark. 
\begin{remark}\label{Remarks4T(M)}Suppose that $G>0$. Then the following assertions hold.\\
{\rm{(a)}} For any $\mathbb F$-semimartingale $X$, we get 
\begin{eqnarray}\label{T(X)}
{\cal T}(X)=X^{\tau}-{\widetilde G}^{-1}I_{\Lbrack0,\tau\Lbrack}\is [X, m].\end{eqnarray}
{\rm{(b)}} For a pair of $\mathbb F$-semimartingales $(X, Y)$, due to $\Delta m=\widetilde G-G_{-}$, we have 
\begin{eqnarray}\label{[X, T(Y)]}
[{\cal T}(X), Y]=[X, {\cal T}(Y)]={{G_{-}}\over{\widetilde G}}\is [X, Y]^{\tau},\quad \Delta {\cal T}(X)={{G_{-}}\over{\widetilde G}}\Delta X I_{\Lbrack0,\tau\Lbrack}.\end{eqnarray}
\end{remark}
For these properties in general (i.e., without the assumption $G>0$) and other related results and their applications as well, we refer the reader to \cite{ChoulliDeng2020}.
 For any $q\in [1,+\infty)$ and a $\sigma$-algebra ${\cal H}$ on $\Omega\times [0,+\infty)$, we define
\begin{equation}\label{L1(PandD)Local}
L^q\left({\cal H}, P\otimes dD\right):=\left\{ X\ {\cal H}\mbox{-measurable}\ \bigg|\ \E[\vert X_{\tau}\vert^q I_{\{\tau<+\infty\}}]<+\infty\right\}.\end{equation}
Throughout the paper, on $\Omega\times [0,+\infty)$, we consider the $\mathbb F$-optional  $\sigma$-field  denoted by ${\cal O}(\mathbb F)$ and  the $\mathbb F$-progressive  $\sigma$-field denoted by $\mbox{Prog}(\mathbb F)$ (i.e., a process $X$ is said to be $\mathbb F$-progresssive if $X$, as  map on $\Omega\times [0,t]$, is ${\cal F}_t\otimes{\cal B}(\mathbb R)$-measurable, for any $t\in (0,+\infty)$, where ${\cal B}(\mathbb R)$ is the Borel $\sigma$-field on $\mathbb R$). 
\subsection{ A martingale representation theorem under $\mathbb G$: Choulli et al. \cite{ChoulliDavelooseVanmaele} }
 In this subsection, we extend a version of the martingale representation result of Choulli et al. \cite{ChoulliDavelooseVanmaele}, see \cite[Theorems 2.17, 2.20, 2.21]{ChoulliDavelooseVanmaele}, to the case of $\mathbb G$-local martingales instead of martingales. This slight extension becomes possible due to the assumption $G>0$.
\begin{theorem}\label{theo4MartingaleDecomposGeneral} Suppose that $G>0$.  Then for any $\mathbb G$-local martingale  $M^{\mathbb G}$,  there exists a unique triplet $(M^{\mathbb F},\varphi^{(o)}, \varphi^{(pr)}) $ satisfying the following properties: $M^{\mathbb F}\in {\cal M}_{0,loc}(\mathbb F)$, $\varphi^{(o)}\in{\cal I}^o_{loc}\left(N^{\mathbb G},\mathbb G\right)$, $\varphi^{(pr)}\in  L^1_{loc}\left({\rm{Prog}}(\mathbb F), P\otimes dD\right)$,   
 \begin{eqnarray}\label{Condition1}
\E\left[\varphi^{(pr)}_{\tau}\ \big|\ {\cal F}_{\tau}\right]I_{\{\tau<+\infty\}}=0,\ \ P\mbox{-a.s.},
\end{eqnarray}
and
\begin{equation}\label{MartingaleDecomposGeneral}
\left(M^{\mathbb G}\right)^{\tau}=M^{\mathbb G}_0+G_{-}^{-2}\is{\cal T}( M^{\mathbb F})
+\varphi^{(o)}\is N^{\mathbb G}+\varphi^{(pr)}\is D.\end{equation}
\end{theorem}

\begin{proof} Let  $M^{\mathbb G}\in {\cal M}_{0,loc}(\mathbb G)$, then there exists a sequence of $\mathbb G$-stopping times that increases to infinity such that $(M^{\mathbb G})^{T_n}$ is a $\mathbb G$-martingale.  On the one hand, due to $G>0$ and \cite[Proposition B.2-(b)]{ACDJ1}, we deduce the existence of $\mathbb F$-stopping times $(\sigma_n)_n$ that increases to infinity and $T_n\wedge\tau=\sigma_n\wedge\tau$ for any $n\geq 1$. On the other hand, by applying \cite[Theorem 2.21]{ChoulliDavelooseVanmaele} to each $(M^{\mathbb G})^{T_n}-(M^{\mathbb G})^{ T_{n-1}}$, we deduce the existence of a unique triplet $(M^{\mathbb F,n},\varphi^{(o,n)}, \varphi^{(pr,n)}) $ belonging to $ {\cal M}_{0,loc}(\mathbb F)\times {\cal I}^o_{loc}\left(N^{\mathbb G},\mathbb G\right)\times  L^1_{loc}\left({\rm{Prog}}(\mathbb F), P\otimes dD\right)$  and satisfying  
$$
\E\left[\varphi^{(pr,n)}_{\tau}\ \big|\ {\cal F}_{\tau}\right]I_{\{\tau<+\infty\}}=0,\ \ P\mbox{-a.s.},
$$
and
$$
(M^{\mathbb G})^{\tau\wedge T_n}-(M^{\mathbb G})^{\tau\wedge T_{n-1}}=G_{-}^{-2}\is{\cal T}( M^{\mathbb F,n})
+\varphi^{(o,n)}\is N^{\mathbb G}+\varphi^{(pr,n)}\is D.$$
Then notice that $(M^{\mathbb G})^{\tau}-M^{\mathbb G}_0=\sum_{n\geq 1} ((M^{\mathbb G})^{\tau\wedge T_n}-(M^{\mathbb G})^{\tau\wedge T_{n-1}})$, and put
\begin{eqnarray*}
&&\sigma_0:=0,\quad \varphi^{(o)}:=\sum_{n\geq 1} I_{\Lbrack \sigma_{n-1},\sigma_n\Lbrack} \varphi^{(o,n)},\  \varphi^{(pr)}:=\sum_{n\geq 1} I_{\Lbrack \sigma_{n-1},\sigma_n\Lbrack} \varphi^{(pr,n)},\\
&&\mbox{and}\quad M^{\mathbb F}:=\sum_{n\geq 1} I_{\Lbrack \sigma_{n-1},\sigma_n\Lbrack}\is M^{\mathbb F,n}.\end{eqnarray*}
This ends the proof of the theorem.\end{proof}
We end this section by  the following lemma. 
\begin{lemma} Let $\sigma$ be an $\mathbb H$-stopping time. $Z$ is a deflator for $(X^{\sigma},\mathbb H)$ if and only if there exists a pair of processes $(K_1, K_2)$ such that $K_1=(K_1)^{\sigma}$, ${\cal E}(K_1)$ is a deflator for  $(X^{\sigma},\mathbb H)$, $K_2$ is an $\mathbb H$-local supermartingale satisfying $(K_2)^{\sigma}\equiv 0$, $\Delta K_2>-1$, and $Z={\cal E}(K_1+K_2)={\cal E}(K_1){\cal E}(K_2)$.
\end{lemma}
The proof of this lemma is straightforward and will be omitted. This lemma shows, in a way or another, that when dealing with the stopped model $(X^{\sigma},\mathbb H)$, there is no loss of generality in focusing only on the part ``up-to-$\sigma$" of the deflator, and assume that the deflator is flat after $\sigma$.
\section{Explicit characterization of all deflators under random horizon}\label{SectionOfDeflators}
This section constitutes one of the principal contributions of the paper, and focuses on explicitly parametrizing the set of all deflators for $(S^{\tau},\mathbb G)$ in terms of deflators for $(S,\mathbb F)$. This section contains three subsections. The main results are stated and interpreted in the first subsection, while the second subsection gives their proofs. The third subsection illustrates the main results on a particular jump-diffusion  and the discrete-time cases. We end this paragraph by useful results on characterization of deflators and local martingale deflators. Throughout the paper,  processes will be compared to each other in the following sense.
\begin{eqnarray}\label{IncPro}
 X \succeq Y\quad \mbox{if}\quad X-Y \quad \mbox{is a nondecreasing process.} \end{eqnarray}
 The following proposition characterizes (local martingale) deflators.
  \begin{proposition}\label{GeneralSupDeflators} 
Let $X$ be an  ${\mathbb H}$-semimartingale and $Z$ be a process. Then the following assertions hold.\\
 {\rm{(a)}} $Z$ is a deflator for $(X, {\mathbb H})$ if and only if there exists a unique pair $(N, V)$ such that $N \in {\cal M}_{loc}(\mathbb H)$, $V$ is nondecreasing RCLL and ${\mathbb H}$-predictable, 
 \begin{eqnarray}
 &&\hskip -1.3cm Z :=Z_0{\cal E}(N){\cal E}(-V), \quad N_0=V_0=0,\quad \Delta N> -1,\quad \Delta V<1,\label{MultiDecompoDeflator}\\
 &&\hskip -1.3cm \sup_{0<s\leq\cdot}\vert\Delta Y^{(\varphi)} \vert\in{\cal A}_{loc}(\mathbb H)\ \&\ {1\over{1-\Delta V}}\is V \succeq  A^{(\varphi, N,\mathbb H)},\ \forall\ \varphi\in \Theta_b(X, \mathbb H).\label{deflatorAssumptions}
  \end{eqnarray}
 Here $Y^{(\varphi)}:= \varphi \is X + [\varphi \is X, N] $ and $A^{(\varphi, N,\mathbb H)}\in{\cal A}_{loc}(\mathbb H)$  is $\mathbb H$-predictable such that 
 $Y^{(\varphi)}-A^{(\varphi, N,\mathbb H)}\in{\cal M}_{loc}(\mathbb H)$. $ {\Theta}_b(X, \mathbb H)$ is the set of bounded $\varphi$ that belongs to $ {\Theta}(X, \mathbb H)$ given by
  \begin{eqnarray}\label{SpaceL(X,H)}
 {\Theta}(X, \mathbb H)  := \left\{ \varphi \mbox{\ is \ } {\mathbb H}\mbox{-predictable \ }: \quad \varphi \Delta X > -1 \right\}.\end{eqnarray}
 {\rm{(b)}} $Z$ is a local martingale deflator for $(X, {\mathbb H})$ (i.e., $Z\in {\cal Z}_{loc}(X,\mathbb H)$) if and only if there exist a real-valued positive and bounded  $\mathbb H$-predictable process $\varphi$ and a unique $N \in {\cal M}_{loc}(\mathbb H)$ such that $N_0=0$,  \begin{eqnarray}
 && \hskip -1cm Z :=Z_0{\cal E}(N),\quad \Delta N> -1,\quad  \sup_{0<s\leq\cdot}\vert \varphi_s \Delta X_s\vert(1 + \Delta N_s) \in {\cal A}_{loc}(\mathbb H),\label{MartingaleDeflator1}\\
 && \hskip -1cm \varphi \is X + [\varphi \is X, N] \in {\cal M}_{loc}(\mathbb H),\label{MartingaleDeflator2}
  \end{eqnarray}

\end{proposition}
The proof of this proposition is relegated to Appendix \ref{AppB}. 
\subsection{Main results}
This subsection states our main results of this section. 
\begin{theorem}\label{GeneralDeflators} 
Suppose that $G > 0$, and let  $Z^{\mathbb G}$ be a process such that $Z^{\mathbb G}=(Z^{\mathbb G})^{\tau}$. Then the following assertions are equivalent.\\
{\rm{(a)}} $Z^{\mathbb G}$ is a deflator for $(S^{\tau}, \mathbb G)$  (i.e., $Z^{\mathbb G}\in {\cal D}(S^{\tau}, \mathbb G)$).\\
{\rm{(b)}} There exists a unique $\left(K^{\mathbb F}, V^{\mathbb F}, \varphi^{(o)}, \varphi^{(pr)}\right)$ such that $K^{\mathbb F}\in {\cal M}_{loc}(\mathbb F)$,  $V^{\mathbb F}$ is an $\mathbb F$-predictable RCLL and nondecreasing process, $\varphi^{(o)}\in {\cal I}^o_{loc}(N^{\mathbb G},\mathbb G)$, $\varphi^{(pr)}\in L^1_{loc}({\rm{Prog}}(\mathbb F),P\otimes dD)$ such that $V^{\mathbb F}_0=K^{\mathbb F}_0=0$,
\begin{eqnarray*}{\cal E}(K^{\mathbb F}){\cal E}(-V^{\mathbb F})\in {\cal D}(S, \mathbb F),\quad E\left[\varphi^{(pr)}_{\tau}\big|{\cal F}_{\tau}\right]I_{\{\tau<+\infty\}}=0\ P\mbox{-a.s.},\end{eqnarray*}
 \begin{equation}\label{ineq1a}
\  {\varphi^{(pr)}}  > -\left[ {{G_{-}}\over{\widetilde G}}(1 + \Delta K^{\mathbb F})  + \varphi^{(o)} {{G}\over{\widetilde G}} \right],\ P\otimes dD-a.e.,
\end{equation}
 \begin{equation}\label{ineq2a}
-{{G_{-}}\over{G}}(1 + \Delta K^{\mathbb F})<\varphi^{(o)}\ P\otimes dD^{o,\mathbb F}\mbox{-a.e.,}\ \varphi^{(o)}\Delta D^{o,\mathbb F} < (1+\Delta K^{\mathbb F})G_{-},
 \end{equation}
\begin{equation}\label{repKG1a}
Z^{\mathbb G}={\cal E}(K^{\mathbb G}){\cal E}(-V^{\mathbb F})^{\tau},\ K^{\mathbb G}={\cal T}({K^{\mathbb F}}-{1\over{ G_{-}}}\is {m})+\varphi^{(o)}\is N^{\mathbb G}+\varphi^{(pr)}\is D.\end{equation}
{\rm{(c)}} There exists a unique triplet $\left(Z^{\mathbb F}, \varphi^{(o)}, \varphi^{(pr)}\right)$ such that $Z^{\mathbb F}\in{\cal D}(S, \mathbb F)$, $(\varphi^{(o)},\varphi^{(pr)})\in{\cal I}^o_{loc}(N^{\mathbb G},\mathbb G)\times L^1_{loc}({\rm{Prog}}(\mathbb F),P\otimes dD)$,
  \begin{eqnarray}
 && \hskip -0.5cm\E[{\varphi}^{(pr)}_{\tau}\ \big|\ {\cal F}_{\tau}]I_{\{\tau<+\infty\}}=0\quad P\mbox{-a.s.}\label{Condition4Fi(pr)}\\
&& \hskip -0.5cm\varphi^{(pr)}>-1,\quad -{{\widetilde G}\over{ G}}<\varphi^{(o)},\quad \varphi^{(o)}(\widetilde G -G)<\widetilde G,\quad P\otimes dD\mbox{-a.e.,} \label{ineqMultiGeneral1}\end{eqnarray}
and 
\begin{equation}\label{repKGMultiGEneral}
Z^{\mathbb G}={{(Z^{\mathbb F})^{\tau}}\over{{\cal E}(G_{-}^{-1}\is m)^{\tau}}}{\cal E}(\varphi^{(o)}\is N^{\mathbb G}){\cal E}(\varphi^{(pr)}\is D).\end{equation}
\end{theorem}

The proof of this theorem is relegated to the second subsection. Herein, we discuss the meaning of this theorem and its possible extension. In virtue of Theorem \ref{theo4MartingaleDecomposGeneral}, see also \cite{ChoulliDavelooseVanmaele} for detailed discussions, the model $(S^{\tau}, \mathbb G)$ has three types of ``orthogonal" risks. Both the {\it pure financial risk and the correlation-risk} resulting from the correlation of $\tau$ with $(S, \mathbb F)$ are represented by $M^{\mathbb F}$, the {\it pure default risk of type one} denoted by $\varphi^{(o)}\is N^{\mathbb G}$,  and the {\it pure default of type two} that takes the form of $\varphi^{(pr)}\is D$. Thus, assertion (c) in the above theorem states that the deflator of $(S^{\tau}, \mathbb G)$ is the product of three deflators that are mutually orthogonal, namely the deflator for the pure financial and correlated risks, the deflator for the pure default risk of type one, and the deflator for pure default risk of type two. Similarly, for assertion (b), one can remark that the stochastic logarithm of a deflator is the sum of the stochastic logarithm of the three types of orthogonal deflators. \\
It is important to mention that both (b) $\Longleftrightarrow$ (c) and (c) $\Longrightarrow$ (a) hold under no assumption, while the condition $G>0$ intervenes in the proof of the converse of the latter implication. However, even though whether $G$ vanishes or not is important for the existence of deflators --see \cite{ACDJ1} for details--,  the role of the condition $G>0$ in Theorem \ref{GeneralDeflators} lies in simplifying the statement of the results and making the proof simpler only. The key point in this simplicity lies in the fact that, under $G>0$, passing from $\mathbb G$-``localization"  to $\mathbb F$-``localization" occurs without any problem, see the proof of Theorem \ref{theo4MartingaleDecomposGeneral} for instance. Indeed, one can drop this assumption --as in \cite{ACDJ1}-- at the expenses of dealing with the family of models $\Bigl\{\left(I_{\{G_{-}\geq\delta\}}\is S-\sum \Delta S I_{\{\widetilde G=0\ \&\ \delta\leq G_{-}\}}, \mathbb F\right):\ \delta\in(0.1)\Bigr\}$ in a way or another instead of one model $(S, \mathbb F)$. Herein, the main difficulty lies somehow in defining the concept of deflator for the whole family of models whose existence should be ``consistent" with the family of models fulfills NUPBR. Once this deflator for the family of models is defined precisely, our theorem remains valid by substituting deflator for $(S, \mathbb F)$ by deflator for the family of models under consideration. \\
One of the key ideas behind the equivalence between assertions (c) and (b) of Theorem \ref{GeneralDeflators}, lies in the following result that is interesting in itself.

\begin{proposition}\label{caseofGlocalmartingale}
   For any $M^{\mathbb F}\in {\cal M}_{loc}(\mathbb F)$, the process $M^{\mathbb G}$ given by 
\begin{eqnarray}\label{SpecialDeflators}
 M^{\mathbb G}:={{(M^{\mathbb F})^{\tau}}\over{{\cal E}(G_{-}^{-1}\is m)^{\tau}}},\end{eqnarray}
 is a $\mathbb G$-local martingale that is orthogonal \footnote{Two $\mathbb H$-local martingale $M_1$ and $M_2$ are said to be orthogonal if $[M_1, M_2]\in {\cal M}_{loc}(\mathbb H)$.} to $\varphi\is N^{\mathbb G}$ and $\psi\is D$, for any pair $(\varphi, \psi)\in {\cal I}_{loc}(N^{\mathbb G}, \mathbb G)\times L_{loc}^1({\rm{Prog}}(\mathbb F),P\otimes dD)$ such that $E[\psi_{\tau}\big|{\cal F}_{\tau}]I_{\{\tau<+\infty\}}=0$ $P$-a.s. and both $[M^{\mathbb F}, \varphi\is N^{\mathbb G}]$ and $[M^{\mathbb F}, \psi\is D]$ belong to ${\cal A}_{loc}(\mathbb G)$.\\
 Furthermore, we always have
 \begin{eqnarray}\label{ExponetialArreter}
 {{{\cal E}(M^{\mathbb F})^{\tau}}\over{{\cal E}(G_{-}^{-1}\is m)^{\tau}}}={\cal E}\Bigl({\cal T}(M^{\mathbb F})-G_{-}^{-1}\is {\cal T}(m)\Bigr).
 \end{eqnarray}
  \end{proposition} 
  \begin{proof} It is clear that $M^{\mathbb G}$ is orthogonal to the $\mathbb G$-local martingales $\varphi\is N^{\mathbb G}$ and $\psi\is D$  provided that $M^{\mathbb G}(\varphi\is N^{\mathbb G})$ and $M^{\mathbb G}(\psi\is D)$   are special $\mathbb G$-semimartingales. Hence, we focus now on proving that $M^{\mathbb G}\in {\cal M}_{loc}(\mathbb G)$. To this end,  remark that due to Yor's formula given by 
  \begin{eqnarray}\label{YorFormula}
  {\cal E}(X){\cal E}(Y)={\cal E}(X+Y+[X,Y]),\ \mbox{for any semimartingales $X$ and $Y$},
  \end{eqnarray}
 we get $1/{\cal E}(X)={\cal E}(-X+(1+\Delta X)^{-1}\is [X,X])$, which holds for any semimartingale $X$ such that $1+\Delta X>0$. Thus, a combination of the latter equality with (\ref{T(X)}) and $\Delta m={\widetilde G}-G_{-}$ leads to   \begin{eqnarray}\label{Exponential1}
{1\over{{\cal E}(G_{-}^{-1}\is m)^{\tau}}}&&={\cal E}\left(-{1\over{G_{-}}}\is m^{\tau}+{{G_{-}^{-2}}\over{1+ G_{-}^{-1}\Delta m}}\is [m,m]^{\tau}\right)\nonumber\\
&&={\cal E}\left(-{1\over{G_{-}}}\is {\cal T}( m)\right).
\end{eqnarray}
Thus, by combining this equality with  (\ref{[X, T(Y)]}) applied to the pair $(M^{\mathbb F}, m)$, the integration by part formula, and (\ref{T(X)}), we get
\begin{eqnarray*}
{{(M^{\mathbb F})^{\tau}}\over{{\cal E}(G_{-}^{-1}\is m)^{\tau}}}&&=M^{\mathbb F}_0-{{ M^{\mathbb F}_{-}}\over{G_{-} {\cal E}_{-}(G_{-}^{-1}\is m)}}\is {\cal T}( m)+{1\over{{\cal E}_{-}(G_{-}^{-1}\is m)}}\is (M^{\mathbb F})^{\tau}\\
&&-{1\over{G_{-}{\cal E}_{-}(G_{-}^{-1}\is m)}}\is[M^{\mathbb F}, {\cal T}(m)]\\
&&=M^{\mathbb F}_0-{{ M^{\mathbb F}_{-}}\over{G_{-} {\cal E}_{-}(G_{-}^{-1}\is m)}}\is {\cal T}( m)+{1\over{{\cal E}_{-}(G_{-}^{-1}\is m)}}\is {\cal T}(M^{\mathbb F}) .\end{eqnarray*} This proves that $M^{\mathbb G}\in {\cal M}_{loc}(\mathbb G)$, and hence the proof of the first assertion is complete. To prove (\ref{ExponetialArreter}), we combine (\ref{Exponential1}), (\ref{YorFormula}), (\ref{[X, T(Y)]}) applied to the pair $(M^{\mathbb F}, m)$ and (\ref{T(X)}), and derive 
\begin{eqnarray*}
 {{{\cal E}(M^{\mathbb F})^{\tau}}\over{{\cal E}(G_{-}^{-1}\is m)^{\tau}}}&&={\cal E}(M^{\mathbb F})^{\tau}{\cal E}\left(-G_{-}^{-1}\is {\cal T}( m)\right)\\
 &&={\cal E}\left( (M^{\mathbb F})^{\tau}-G_{-}^{-1}\is {\cal T}( m)-[M^{\mathbb F}, G_{-}^{-1}\is {\cal T}( m)]\right)\\
 &&={\cal E}\Bigl({\cal T}(M^{\mathbb F})-G_{-}^{-1}\is {\cal T}(m)\Bigr).
\end{eqnarray*}
 This ends the proof of the proposition.\end{proof}
As a particular case of Theorem  \ref{GeneralDeflators}, we characterize the set of all local martingale deflators for $(S^{\tau},\mathbb G)$, denoted by ${\cal Z}_{loc}(S^{\tau},\mathbb G)$, as follows.

\begin{theorem}\label{LocalMartingaleDeflator}
Suppose $G > 0$, and let  $K^{\mathbb G}$ be a $\mathbb G$-semimartingale such that $K^{\mathbb G}=(K^{\mathbb G})^{\tau}$. Then the following assertions are equivalent.\\
{\rm{(a)}} $Z^{\mathbb G}:={\cal E}\left(K^{\mathbb G}\right)$ is a local martingale deflator for $(S^{\tau}, \mathbb G)$.\\
{\rm{(b)}} There exists a unique triplet $(K^{\mathbb F},\varphi^{(o)}),\varphi^{(pr)})$ that belongs to the set ${\cal M}_{loc}(\mathbb F)\times {\cal I}^o_{loc}(N^{\mathbb G},\mathbb G)\times L^1_{loc}({\rm{Prog}}(\mathbb F),P\otimes dD)$ such that $K^{\mathbb F}_0=0$, both conditions (\ref{ineq1a}) and (\ref{ineq2a}) hold, $E[\varphi^{(pr)}_{\tau}\ \big|\ {\cal F}_{\tau}]I_{\{\tau<+\infty\}}=0$ $P$-a.s., ${\cal E}\left(K^{\mathbb F} \right)$ belongs to ${\cal Z}_{loc}(S,\mathbb F),$ and 
\begin{equation}\label{repKG1}
K^{\mathbb G}={\cal T}({K^{\mathbb F}}) - G_{-}^{-1}\is{\cal T}({m})+\varphi^{(o)}\is N^{\mathbb G}+\varphi^{(pr)}\is D.\end{equation}
{\rm{(c)}} There exists a unique triplet $ (Z^{\mathbb F},\varphi^{(o)}, \varphi^{(pr)})$ that belongs to the set ${\cal Z}_{loc}(S,\mathbb F)\times{\cal I}^o_{loc}(N^{\mathbb G},\mathbb G)\times  L^1_{loc}({\rm{Prog}}(\mathbb F),P\otimes dD)$ and satisfies the three conditions (\ref{Condition4Fi(pr)}), (\ref{ineqMultiGeneral1}) and
\begin{eqnarray}\label{repKGMulti}
Z^{\mathbb G}={{(Z^{\mathbb F})^{\tau}}\over{{\cal E}(G_{-}^{-1}\is m)^{\tau}}}{\cal E}(\varphi^{(o)}\is N^{\mathbb G}){\cal E}(\varphi^{(pr)}\is D).
\end{eqnarray}
\end{theorem}

\begin{proof} It is clear that the equivalence between assertions (b) and (c) follows the same arguments as those in the proof of (b) $\Longleftrightarrow$ (c) of Theorem \ref{GeneralDeflators} by putting $V^{\mathbb F}\equiv 0$. Similarly, the proof of (a) $\Longrightarrow$ (b) follows the same footsteps of the proof of (a) $\Longrightarrow$ (b) of Theorem \ref{GeneralDeflators}, by putting $V^{\mathbb G}\equiv 0$, replacing the order sign $\preceq$ by $=$, and letting $\varphi$ to be real-valued positive and bounded  predictable only instead of belonging to ${\Theta}(S^{\tau}, \mathbb G)$. \\
Finally, to prove (c) $\Longrightarrow$ (a), we remark that a combination of Proposition \ref{caseofGlocalmartingale} and the fact that  there exists a positive and bounded $\mathbb F$-predictable process $\varphi$ such that $Z^{\mathbb F}(\varphi\is S)\in {\cal M}_{loc}(\mathbb F)$ (which is due to $Z^{\mathbb F}\in {\cal Z}_{loc}(S, \mathbb F)$), leads to $Z^{\mathbb G}(\varphi\is S^{\tau})$ being a $\mathbb G$-local martingale. This proves assertion (a), and the proof of the theorem is completed.
\end{proof} 

Thanks to Proposition \ref{caseofGlocalmartingale}  and Theorem \ref{LocalMartingaleDeflator}, we single out a particular subclass of local martingale deflators for $(S^{\tau},\mathbb G)$. 
\begin{corollary}\label{caseofGmartingaleDensities}
   If $Z^{\mathbb F}\in {\cal Z}_{loc}(S,\mathbb F)$, then $Z^{\mathbb G}:=(Z^{\mathbb F})^{\tau}/{\cal E}(G_{-}^{-1}\is m)^{\tau}$ belongs to ${\cal Z}_{loc}(S^{\tau}, \mathbb G)$.
  \end{corollary} 
\subsection{Proof of Theorem \ref{GeneralDeflators}}
The proof of this theorem is based on two lemmas. The first lemma characterizes the deflators of $(S^{\tau}, \mathbb G)$ with $\mathbb F$-adapted processes. To this end, we recall that 
\begin{eqnarray}\label{Sequation}
S=S_0+M+A+\sum\Delta S I_{\{\vert \Delta S\vert>1\}},\end{eqnarray}
where $M$ is an $\mathbb F$-locally bounded local martingale and $A\in {\cal A}_{loc}(\mathbb F)$ is an $\mathbb F$-predictable process such that $M_0=A_0=0$.
\begin{lemma}\label{lemma1fortheorem3.2} If $Z^{\mathbb G}\in {\cal D}(S^{\tau}, \mathbb G)$, then there exists a quadruple $(K^{\mathbb F}, V,\varphi^{(o)}, \varphi^{(pr)})$ that belongs to  ${\cal M}_{loc}(\mathbb F)\times {\cal A}_{loc}(\mathbb F)\times {\cal I}^o_{loc}(N^{\mathbb G},\mathbb G)\times L^1_{loc}({\rm{Prog}}(\mathbb F),P\otimes dD)$  and satisfies $K^{\mathbb F}_0=V_0=0$, (\ref{Condition4Fi(pr)}) holds, $V$ is nondecreasing and $\mathbb F$-predictable,  
\begin{eqnarray}
&&Z^{\mathbb G}=Z^{\mathbb G}_0{\cal E}(K^{\mathbb G}){\cal E}(-V)^{\tau},\quad \Delta V<1,\quad \Delta K^{\mathbb G}>-1,\label{equa1forlemma1}\\
&&K^{\mathbb G}:= {\cal T}\left(K^{\mathbb F}-{1\over{G_{-}}}\is m\right)+\varphi^{(o)}\is N^{\mathbb G}+\varphi^{(pr)}\is D,\label{equa2forlemma1}
\end{eqnarray}
and for any bounded $\varphi\in {\Theta}(S,\mathbb F)$ we have
\begin{eqnarray}
&&U^{(\varphi)}:=\sum \varphi\Delta S I_{\{\vert \Delta S\vert>1\}}(1+\Delta K^{\mathbb F}) \in {\cal A}_{loc}(\mathbb F)\label{equa3forlemma1}\\
\mbox{and}\ &&{1\over{1-V}}\is V^{\tau} \succeq I_{\Lbrack0,\tau\Lbrack}\is A^{(\varphi, K^{\mathbb F}, \mathbb F)},\label{equa4forlemma1}\\
\mbox{where}\ &&A^{(\varphi, K^{\mathbb F}, \mathbb F)}:=(U^{(\varphi)})^{p,\mathbb F}+ \varphi\is \langle M, K^{\mathbb F}\rangle^{\mathbb F}+  \varphi\is A.\label{equa5forlemma1}\end{eqnarray}
\end{lemma}
\begin{proof} Let $Z^{\mathbb G}$ is a deflator for $(S^{\tau}, \mathbb G)$. Then, due to Proposition \ref{GeneralSupDeflators}-(a), we get the existence of $K^{\mathbb G}\in {\cal M}_{loc}(\mathbb G)$, and a $\mathbb G$-predictable RCLL and nondecreasing process  $V^{\mathbb G}$ such that  $V^{\mathbb G}_0=K^{\mathbb G}_0=0$, 
 \begin{eqnarray}
  &&Z^{\mathbb G} = Z^{\mathbb G}_0{\cal E}(K^{\mathbb G}){\cal E}(-V^{\mathbb G}), \quad \Delta K^{\mathbb G}> -1, \quad  \Delta V^{\mathbb G}<1, \label{Condition1g}\\
 && \sup_{0<s\leq\cdot}\vert  \varphi_s \Delta S^{\tau}_s\vert(1+\Delta K^{\mathbb G}_s)\in {\cal A}^+_{loc}(\mathbb G) ,\label{Condition2g}\\
 &&{1\over{1-\Delta V^{\mathbb G}}}\is  V^{\mathbb G} \succeq  A^{(\varphi, K^{\mathbb G}, \mathbb G)},\ \mbox{for any}\ \varphi \in {\Theta}_b(S^{\tau}, \mathbb G).\hskip 1cm\label{Condition3g}\end{eqnarray}
 where $A^{(\varphi, K^{\mathbb G}, \mathbb G)}\in {\cal A}_{loc}(\mathbb G)$ is $\mathbb G$-predictable such that 
 \begin{eqnarray}\label{AGK}
\varphi \is S^{\tau} + [ \varphi \is S^{\tau}, K^{\mathbb G}]-A^{(\varphi, K^{\mathbb G}, \mathbb G)}\in {\cal M}_{loc}(\mathbb G).\end{eqnarray}
 A direct application of Lemma \ref{PortfolioGtoF}-(c)  to $V^{\mathbb G}$ implies the existence of an $\mathbb F$-predictable and nondecreasing process $V\in {\cal A}_{loc}(\mathbb F)$ such that  $\Delta V<1$ and $ V^{\mathbb G}=V^{\tau}$. Thus, by inserting this in  (\ref{Condition1g}), (\ref{equa1forlemma1})  follows immediately. \\
 By applying Theorem \ref{theo4MartingaleDecomposGeneral} to $K^{\mathbb G}$, we get the existence of $(N^{\mathbb F},\varphi^{(o)}, \varphi^{(pr)})$ that belongs to ${\cal M}_{loc}(\mathbb F)\times {\cal I}^o_{loc}(N^{\mathbb G},\mathbb G)\times L^1_{loc}({\rm{Prog}}(\mathbb F),P\otimes dD)$  such that $N^{\mathbb F}_0=0$, $\E[\varphi^{(pr)}_{\tau}\ \big|\ {\cal F}_{\tau}]I_{\{\tau<+\infty\}}=0$ $P\mbox{-a.s.}$,  and 
\begin{eqnarray}\label{RepKG4Proof}
K^{\mathbb G}= \frac{1}{G_-^2}  \is {\cal T}({N^{\mathbb F}})+\varphi^{(o)}\is N^{\mathbb G}+\varphi^{(pr)}\is D.\end{eqnarray}
Thus, on the one hand, by putting 
\begin{eqnarray}\label{K4F}
 K^{\mathbb F}:={1\over{G_{-}}}\is m+{1\over{G_{-}^2}}\is N^{\mathbb F}\in {\cal M}_{loc}(\mathbb F)\end{eqnarray}
 and inserting it in (\ref{RepKG4Proof}), we obtain (\ref{equa2forlemma1}).  On the other hand, we consider a bounded $ \varphi \in {\Theta}(S, \mathbb F)$, and remark that $ \varphi\is S^{\tau} +  \varphi\is [K^{\mathbb G} , S^{\tau}]$ is a special $\mathbb G$-semimartingale due to the condition (\ref{Condition2g}). By combining (\ref{Sequation}), (\ref{T(X)}), (\ref{[X, T(Y)]}), (\ref{RepKG4Proof}),  and the fact that the three processes $\varphi\is[A^{\tau},K^{\mathbb G}]$ \footnote{This is due to Yoeurp's lemma, see \cite[Th\'eor\`eme 36, Chapter VII, p. 245]{DellacherieMeyer80}.}, $\varphi\Delta M{\varphi}^{(o)}\is N^{\mathbb G}$ and $\varphi\Delta M{\varphi}^{(pr)}\is D$ \footnote{These two processes belong to ${\cal M}_{loc}(\mathbb G)$ is due to a combination of $\varphi\Delta M$ being bounded and $\mathbb F$-optional with \cite[Proposition 2.13]{ChoulliDavelooseVanmaele}.}are $\mathbb G$-local martingales, we derive 
  \begin{eqnarray*}
 &&\varphi\is S^{\tau} +  \varphi\is [K^{\mathbb G} , S^{\tau}] \\
 &=& \varphi\is M^{\tau}+\varphi\is A^{\tau}+ \sum \varphi\Delta S^{\tau} I_{\{\vert \Delta S\vert>1\}} +\varphi\is [K^{\mathbb G} , S^{\tau}],\\
&=&  {\mathbb G}\mbox{-local martingale}+{{\varphi}\over{\widetilde G}}I_{\Lbrack0,\tau\Lbrack}\is [M, m] +  \varphi\is A^{\tau} + {{\varphi}\over{G_{-}\widetilde G}}I_{\Lbrack0,\tau\Lbrack}\is [N^{\mathbb F}, M]\\
&+& \sum \varphi\Delta S I_{\{\vert \Delta S\vert>1\}}\left[(1+{{\Delta N^{\mathbb F}}\over{G_{-}\widetilde G}})I_{\Lbrack0,\tau\Lbrack}+\varphi^{(o)}\Delta N^{\mathbb G}+\varphi^{(pr)}\Delta D \right].
\end{eqnarray*}
Thus, in virtue of Lemma \ref{lemmaV}-(a), we conclude that (\ref{Condition2g}) is equivalent to
\begin{equation}\label{Wintegrable}
W:=\sum \varphi\Delta S I_{\{\vert \Delta S\vert>1\}}[(1+{{\Delta N^{\mathbb F}}\over{G_{-}\widetilde G}})I_{\Lbrack0,\tau\Lbrack}+\varphi^{(o)}\Delta N^{\mathbb G}+\varphi^{(pr)}\Delta D ]\end{equation}
being $\mathbb G$-locally integrable, and
\begin{eqnarray}\label{AKG1}
A^{(\varphi, K^{\mathbb G}, \mathbb G)}=\varphi\is A^{\tau}+{\varphi}I_{\Lbrack0,\tau\Lbrack}\is \langle M, K^{\mathbb F} \rangle^{\mathbb F} +W^{p,\mathbb G}.\end{eqnarray}
Remark that $0<1+\Delta K^{\mathbb G}=(1+{{\Delta N^{\mathbb F}}\over{G_{-}\widetilde G}})I_{\Lbrack0,\tau\Lbrack}+\varphi^{(o)}\Delta N^{\mathbb G}+\varphi^{(pr)}\Delta D $ is due to the second condition in (\ref{Condition1g}), (\ref{RepKG4Proof}), and (\ref{[X, T(Y)]}). Thus, in virtue of Lemma \ref{Lemma4Wprocess} and (\ref{K4F}), we conclude that $W\in {\cal A}_{loc}(\mathbb G)$ iff both processes 
\begin{eqnarray}\label{W1}
W^{(1)}:=\sum \varphi\Delta S I_{\{\vert \Delta S\vert>1\}}[(1+{{\Delta N^{\mathbb F}}\over{G_{-}\widetilde G}})]I_{\Lbrack0,\tau\Lbrack}={{G_{-}}\over{\widetilde G}}I_{\Lbrack0,\tau\Lbrack}\is U^{(\varphi)}
\end{eqnarray} and $W^{(2)}:=\sum \varphi\Delta S I_{\{\vert \Delta S\vert>1\}}[\varphi^{(o)}\Delta N^{\mathbb G}+\varphi^{(pr)}\Delta D] $ belong to $ {\cal A}_{loc}(\mathbb G)$. \\
It is clear that  $W^{(2)}$ belongs to $ {\cal A}_{loc}(\mathbb G)$ if and only if it is a $\mathbb G$-local martingale, and it is also clear that $W^{(1)}\in {\cal A}_{loc}(\mathbb G)$ if and only if (\ref{equa3forlemma1}) holds, see Lemma \ref{lemmaV}-(b). Hence, due to Lemma \ref{lemmaV}-(a) again, we get 
$$W^{p,\mathbb G}=(W^{(1)})^{p,\mathbb G}=I_{\Lbrack0,\tau\Lbrack}\is \left(U^{(\varphi)}\right)^{p,\mathbb F}.$$
Therefore, by combining this equality with (\ref{AKG1}), (\ref{Condition3g}), and $V^{\mathbb G}=V^{\tau}$, the condition (\ref{equa4forlemma1}) follows immediately.  This ends the proof of the lemma.\end{proof}
\begin{lemma}\label{lemma2fortheorem3.2} Let $(K^{\mathbb F},\varphi^{(o)}, \varphi^{(pr)})\in {\cal M}_{loc}(\mathbb F)\times {\cal I}^o_{loc}(N^{\mathbb G},\mathbb G)\times L^1_{loc}({\rm{Prog}}(\mathbb F),P\otimes dD)$  and $K^{\mathbb G}\in {\cal M}_{loc}(\mathbb G)$ such that (\ref{Condition4Fi(pr)}) and (\ref{equa2forlemma1}) hold. Then $\Delta K^{\mathbb G}>-1$ if and only if $\Delta K^{\mathbb F}>-1$ and both  (\ref{ineq1a}) d an(\ref{ineq2a}) hold.
\end{lemma}
\begin{proof}Put $\Gamma:=G_{-}{\widetilde G}^{-1}(1+\Delta{K^{\mathbb F}})  - 1,$ and derive 
\begin{eqnarray*}
\Delta K^{\mathbb G}&=&\Delta{\cal T}({K^{\mathbb F}}) - G_{-}^{-1}\Delta{\cal T}({m})+\varphi^{(o)}\Delta N^{\mathbb G}+\varphi^{(pr)}\Delta D\\
&=&\left[\Gamma- \varphi^{(o)} \frac{\Delta D^{o,\mathbb F}}{\widetilde{G}} \right] I_{\Lbrack 0,\tau\Rbrack} +\left[\Gamma+\varphi^{(pr)} + \varphi^{(o)}\frac{G}{\widetilde{G}}\right] I_{\Rbrack \tau\Lbrack}.
\end{eqnarray*}
Therefore,  ${\cal E}(K^{\mathbb G})>0$ iff $1+\Delta K^{\mathbb G}>0$ if and only if 
\begin{eqnarray}
\Lbrack 0,\tau\Rbrack&\subset& \left\{\frac{G_{-}}{\widetilde{G}} (1 + \Delta K^{\mathbb F})   - \varphi^{(o)} \frac{\Delta D^{o,\mathbb F}}{\widetilde{G}} > 0 \right\},\label{positivityKG1}\\
 \mbox{and}\quad \Rbrack \tau\Lbrack&\subset&\left\{\frac{G_{-}}{\widetilde{G}} (1 + \Delta K^{\mathbb F})  + \varphi^{(o)} \frac{G}{\widetilde{G}} + \varphi^{(pr)}  > 0\right\}.\label{positivityKG2}
\end{eqnarray}
Thus, by putting $\Sigma_1:=\left\{{G_{-}}{\widetilde{G}}^{-1} (1 + \Delta K^{\mathbb F})   - \varphi^{(o)} {\widetilde{G}} ^{-1}{\Delta D^{o,\mathbb F}}> 0 \right\}\cap \Lbrack0,+\infty\Rbrack$, (\ref{positivityKG1}) is equivalent to  $I_{ \Lbrack 0,\tau\Rbrack}\leq I_{\Sigma_1}$. Hence, by taking the $\mathbb F$-optional projection on both sides of this inequality, we get $0<G\leq I_{\Sigma_1}$ on $\Lbrack0,+\infty\Rbrack$. This proves the right inequality in (\ref{ineq2a}). Notice that (\ref{positivityKG2}) is equivalent to 
$${G_{-}}{\widetilde{G}} ^{-1}(1 + \Delta K^{\mathbb F})  + \varphi^{(o)} {G}{\widetilde{G}}^{-1} + \varphi^{(pr)}  > 0, \quad P\otimes dD-a.e.,$$
and (\ref{ineq1a}) is proved.  Now, we focus on proving that $1+\Delta K^{\mathbb F}>0$ and the left inequality in (\ref{ineq2a}). Thanks to $\E_{P\otimes dD}[ \varphi^{(pr)}|\mathcal{O}(\mathbb F)] =0$  $P\otimes dD-a.e.$,  by taking conditional expectation under $P\otimes dD$ with respect to $\mathcal{O}(\mathbb F)$ on both sides of the above inequality, we get
 \begin{equation}\label{InqDe1}
\Sigma:= {G_{-}} (1 + \Delta K^{\mathbb F})  + \varphi^{(o)} G > 0,\quad P\otimes dD-a.e..
  \end{equation}
Remark that this is equivalent to the left inequality in (\ref{ineq2a}), as $\Sigma$ is $\mathbb F$-optional, and hence the proof of  (\ref{ineq2a}) is completed. By taking the $\mathbb F$-optional projection in both sides of $I_{\Rbrack \tau\Lbrack}\leq I_{\{\Sigma>0\}}$ which is equivalent to (\ref{InqDe1}), we get $\Delta D^{o,\mathbb F} \leq I_{ \{   \Sigma > 0 \}}.$ Therefore, we derive
\begin{equation}\label{InEq4}
\{ \Delta D^{o,\mathbb F} >0 \} \subseteq \{  {G_{-}} (1 + \Delta K^{\mathbb F})  > - \varphi^{(o)} G \}.
\end{equation}
On the one hand, due to the right inequality in (\ref{ineq2a}), we deduce that on $\{ \Delta D^{o,\mathbb F} = 0\},$ we have $1+\Delta K^{\mathbb F}>0$. On the other hand, using (\ref{InEq4}) and the right inequality in (\ref{ineq2a}) afterwards again, we get
$$
 \{ 1+\Delta K^{\mathbb F} \leq 0\} \cap \{ \Delta D^{o,\mathbb F} > 0\} \subseteq \{\varphi^{(o)} > 0 , 1+\Delta K^{\mathbb F} \leq 0 , \Delta D^{o,\mathbb F} > 0\} =\emptyset,
$$
or equivalently $\{ \Delta D^{o,\mathbb F} > 0\} \subseteq  \{ 1+\Delta K^{\mathbb F} > 0\}$. Thus, $  1+\Delta K^{\mathbb F} >0$ , and the proof of the lemma is completed.
\end{proof}
\begin{proof}{\it of Theorem \ref{GeneralDeflators}}
 This proof will be achieved in three steps, where we prove the implications (a)$\Longrightarrow$(b), (b)$\Longrightarrow$(c), and (c)$\Longrightarrow$(a) respectively.\\
{\bf Step 1.} Herein, we prove (a)$\Longrightarrow$(b). Supppose that $Z^{\mathbb G}\in {\cal D}(S^{\tau}, \mathbb G)$. Then by applying Lemmas \ref{lemma1fortheorem3.2} -\ref{lemma2fortheorem3.2}, we deduce the existence of  $(K^{\mathbb F}, V,\varphi^{(o)}, \varphi^{(pr)})\in {\cal M}_{loc}(\mathbb F)\times {\cal A}){loc}(\mathbb F)\times {\cal I}^o_{loc}(N^{\mathbb G},\mathbb G)\times L^1_{loc}({\rm{Prog}}(\mathbb F),P\otimes dD)$ satisfying all conditions of assertions (b) except ${\cal E}(K^{\mathbb F}){\cal E}(-V)\in {\cal D}(S, \mathbb F)$. To prove this,  on the one hand, we take the $\mathbb F$-dual predictable projection on both sides of (\ref{equa4forlemma1}) and use $(H\is U)^{p,\mathbb F}={^{p,\mathbb F}(H)}\is U$ (that holds for any bounded process $H$ and any $\mathbb F$-predictable process $U\in {\cal A}_{loc}(\mathbb F)$) to get
\begin{eqnarray}\label{AKFdomination}
{1\over{1-V}}\is V\succeq \left(U^{(\varphi)}\right)^{p,\mathbb F}+ \varphi\is \langle M, K^{\mathbb F}\rangle^{\mathbb F}+  \varphi\is A=A^{(\varphi, K^{\mathbb F}, \mathbb F)}.\end{eqnarray}
On the other hand, we remark that $A^{(\varphi, K^{\mathbb F}, \mathbb F)}\in {\cal A}_{loc}(\mathbb F)$  is the unique $\mathbb F$-predictable process such that $\varphi\is S+[\varphi\is S, K^{\mathbb F}]- A^{(\varphi, K^{\mathbb F}, \mathbb F)}\in {\cal M}_{loc}(\mathbb F)$.  Therefore, by combining this latter fact with (\ref{AKFdomination}) and It\^o's formula  applied to ${\cal E}(K^{\mathbb F}){\cal E}(-V){\cal E}(\varphi\is S)$, we deduce that this process is an $\mathbb F$-supermartingale. This is equivalent to ${\cal E}(K^{\mathbb F}){\cal E}(-V)\in {\cal D}(S,\mathbb F)$, as $\varphi$ is arbitrary bounded element of $\Theta(S,\mathbb F)$, and step 1 is complete.\\
{\bf Step 2.} This step proves (b)$\Longrightarrow$(c). Hence, we suppose that assertion (b) holds. Then there exists $\left(K^{\mathbb F}, V^{\mathbb F}, \overline\varphi^{(o)},\overline\varphi^{(pr)}\right)$ such that $Z^{\mathbb F}:={\cal E}(K^{\mathbb F}){\cal E}(-V^{\mathbb F})$ belongs to ${\cal D} (S, \mathbb F),$  $\overline\varphi^{(o)}\in {\cal I}^o_{loc}(N^{\mathbb G},\mathbb G)$, $\overline\varphi^{(pr)}\in L^1_{loc}({\rm{Prog}}(\mathbb F),P\otimes dD)$, $\E[\overline{\varphi}^{(pr)}_{\tau}\ \big|\ {\cal F}_{\tau}]I_{\{\tau<+\infty\}}=0$ $P\mbox{-a.s.}$, and (\ref{ineq1a})-(\ref{ineq2a})-(\ref{repKG1a}) hold. Put 
\begin{eqnarray}
&&Y:={\cal T}({K^{\mathbb F}}) - G_{-}^{-1}\is{\cal T}({m}),\quad X:=Y+\overline\varphi^{(o)}\is N^{\mathbb G}\nonumber\\
&& \varphi^{(o)}:={{{\widetilde G}\overline\varphi^{(o)}}\over{G_{-}(1+\Delta K^{\mathbb F})}},\quad \varphi^{(pr)}:={{{\widetilde G}\overline\varphi^{(pr)}}\over{G_{-}(1+\Delta K^{\mathbb F})+\overline{\varphi}^{(o)}G}}.\label{adjustPhi}
\end{eqnarray}
As the pair $(\overline\varphi^{(o)},\overline\varphi^{(pr)})$ satisfies (\ref{ineq1a})-(\ref{ineq2a}), it is clear that the pair $(\varphi^{(o)},\varphi^{(pr)})$ satisfies (\ref{ineqMultiGeneral1}). Furthermore, we put as before, 
$$\Gamma:=G_{-}{\widetilde G}^{-1}(1 + \Delta K^{\mathbb F})-1,\quad \widetilde\Omega:=\Omega\times[0,+\infty),$$ and calculate 
\begin{eqnarray*}
1+\Delta X&&= \left[\Gamma+1 -\Delta D^{o,\mathbb F} {{ \overline\varphi^{(o)}}\over{\widetilde G}}\right]I_{\Lbrack0,\tau\Rbrack}+ \left[\Gamma+1 +\overline\varphi^{(o)} {{G}\over{\widetilde G}}\right]I_{\Rbrack\tau\Lbrack}+I_{\widetilde\Omega\setminus \Lbrack0,\tau\Lbrack}>0.\\
1+\Delta Y&&={{G_{-}}\over{\widetilde G}}(1 + \Delta K^{\mathbb F})I_{\Lbrack0,\tau\Lbrack}+I_{\Lbrack-\infty,0\Lbrack\cup\Lbrack\tau,+\infty\Rbrack}>0.
\end{eqnarray*}
Thanks to (\ref{YorFormula}), we derive
$${\cal E}(X_1+X_2)={\cal E}(X_1){\cal E}\left(X_2-(1+\Delta X_1)^{-1}\is [X_1,X_2]\right),$$
for any semimartingales $X_1, X_2\ \mbox{with}\ 1+\Delta X_1>0$. By applying this formula repeatedly, and using $ \varphi^{(o)}=\overline\varphi^{(o)}/(1+\Delta Y)$ and $\varphi^{(pr)}=\overline\varphi^{(pr)}/(1+\Delta X)$ $P\otimes dD$-a.e. which follow directly from (\ref{adjustPhi}), we obtain 
\begin{eqnarray*}
&&{\cal E}(K^{\mathbb G})={\cal E}(X+\overline\varphi^{(pr)}\is D)={\cal E}(X){\cal E}({{\overline\varphi^{(pr)}}\over{1+\Delta X}}\is D)={\cal E}(X){\cal E}(\varphi^{(pr)}\is D)\\
&&={\cal E}(Y){\cal E}({{\overline\varphi^{(o)}}\over{1+\Delta Y}}\is N^{\mathbb G}){\cal E}(\varphi^{(pr)}\is D)={\cal E}(Y){\cal E}(\varphi^{(o)}\is N^{\mathbb G}){\cal E}(\varphi^{(pr)}\is D).
\end{eqnarray*}
Therefore, the equality (\ref{repKGMultiGEneral}) follows immediately from combining the above equality with (\ref{ExponetialArreter}) (see Proposition \ref{caseofGlocalmartingale}). This ends the second step.\\
{\bf Step 3.}  Herein, we deal with (c)$\Longrightarrow$(a). Thus, we suppose that assertion (c) holds, and deduce the existence of a triplet $\left(Z^{\mathbb F}, \varphi^{(o)}, \varphi^{(pr)}\right)$ that belongs to  ${\cal D}(S, \mathbb F)\times {\cal I}^o_{loc}(N^{\mathbb G},\mathbb G)\times L^1_{loc}({\rm{Prog}}(\mathbb F),P\otimes dD)$ satisfying (\ref{Condition4Fi(pr)}), (\ref{ineqMultiGeneral1}),  and      
\begin{eqnarray}\label{ZG2ZF}
Z^{\mathbb G}={{({Z^{\mathbb F}})^{\tau}}\over{{\cal E}(G_{-}^{-1}\is m)^{\tau}}} {\cal E}({\varphi}^{(o)}\is N^{\mathbb G}){\cal E}({\varphi}^{(pr)}\is D).
 \end{eqnarray}
 Then for any bounded $ \varphi \in {\Theta}(S, \mathbb F) $, ${Z^{\mathbb F}}{\cal E}(\varphi\is S)$ is a positive $\mathbb F$-supermartingale. Hence, thanks to  \cite[Theorem 8.21, Chapter II]{JS03}, we deduce the existence of $N \in {\cal M}_{loc}(\mathbb F)$ and an ${\mathbb F}$-predictable RCLL and non decreasing process $V$ such that  $V_0=N_0=0$, $\Delta V<1$ and 
$${Z^{\mathbb F}}{\cal E}(\varphi\is S) = {\cal E}(N){\cal E}(-V).$$ Thus, by inserting this in (\ref{ZG2ZF}), we deduce that, for any $ \varphi \in {\Theta}_b(S, \mathbb F) $,\begin{eqnarray*}
     Z^{\mathbb G}{\cal E}(\varphi \is S)^{\tau}={{{\cal E}(N)^{\tau}}\over{{\cal E}(G_{-}^{-1}\is m)^{\tau}}} {\cal E}({\varphi}^{(o)}\is N^{\mathbb G}){\cal E}({\varphi}^{(pr)}\is D){\cal E}(-V)^{\tau}. \end{eqnarray*}
 In virtue of Proposition \ref{caseofGlocalmartingale}, which allows us to conclude that the process $({{{\cal E}(N)^{\tau}/{\cal E}(G_{-}^{-1}\is m)^{\tau}}}) {\cal E}({\varphi}^{(o)}\is N^{\mathbb G}){\cal E}({\varphi}^{(pr)}\is D)$ is a $\mathbb G$-local martingale, we deduce that  $Z^{\mathbb G}{\cal E}(\varphi \is S)^{\tau}$ is a $\mathbb G$-supermartingale, for any bounded $ \varphi \in {\Theta}(S, \mathbb F) $. Then assertion (a) follows immediately from combining this with the fact that, for any bounded  $ \varphi^{\mathbb G}\in{\Theta}(S^{\tau}, \mathbb G) $, there exists a bounded  $ \varphi^{\mathbb F} \in {\Theta}(S, \mathbb F) $ such that $ \varphi^{\mathbb G} = \varphi^{\mathbb F}$ on $\Lbrack0,\tau\Lbrack$ (see Lemma \ref{PortfolioGtoF}-(a) for details). This ends the proof of the theorem.\end{proof}
\subsection{Two particular cases for  $(S,\mathbb F)$: The jump-diffusion and discrete-time}\label{section5}
In this section, we illustrate the main results on two particular cases. Precisely, we discuss the case when $(S,\mathbb F)$ follows a jump-diffusion model, and the case of discrete time model for $(S,\mathbb F)$. 
\subsubsection{The case of jump-diffusion for $(S,\mathbb F)$}
This subsection focuses on a particular case of a jump-diffusion process for the market model $(S,\mathbb F)$. Herein, we suppose that a standard Brownian motion $W$ and a Poisson process $N$ with intensity $\lambda>0$ are defined on  the probability space $(\Omega, {\cal F}, P)$ and are  independent. Let $\mathbb F$ be the completed and
 right continuous filtration generated by $W$ and $N$. The stock's price process is supposed to have the following dynamics
\begin{equation}\label{SPoisson2}
S_t=S_0 {\cal E} (X)_t,  \  X_t = \sigma\is W_t+  \zeta\is {N}^{\mathbb F}_t + \int_{0}^{t} \mu_s ds, \  {N_t}^{\mathbb F}:=N_t-\lambda t.
\end{equation}
Here $\mu$, $\sigma$ and $\zeta$ are bounded  and $\mathbb F$-predictable processes, and there exists a constant $\delta\in(0,+\infty)$ such that
\begin{equation}\label{parameters2}
 \sigma>0,\quad \zeta>-1,\quad \sigma+\vert\zeta\vert\geq \delta,\ P\otimes dt\mbox{-a.e.}.
 \end{equation}
\begin{theorem}\label{JumpDiffDeflator2} Let $Z^{\mathbb G}:={\cal E}\left(K^{\mathbb G}\right)$ be a positive $\mathbb G$-local martingale. Suppose $G>0$ and $S$ is given by (\ref{SPoisson2})-(\ref{parameters2}). Then the following are equivalent.\\
{\rm{(a)}} $Z^{\mathbb G}$ is a local martingale deflator for $(S^{\tau}, \mathbb G)$,\\
{\rm{(b)}}  There exist $(\psi_1,\psi_2)\in  L^1_{loc}(W,\mathbb F)\times L^1_{loc}(N^{\mathbb F},\mathbb F)$ ,  $\varphi^{(o)}\in{\cal I}^o_{loc}(N^{\mathbb G},\mathbb G)$  and $\varphi^{(pr)}\in L^1_{loc}( \rm{Prog}(\mathbb F),P\otimes dD)$ satisfying (\ref{Condition4Fi(pr)}) and the following
\begin{eqnarray*}\label{repKG1discrete}
 && K^{\mathbb G}=\psi_1\is {\cal T}(W) + \psi_2 \is{\cal T}( N^{\mathbb F})-G_{-}^{-1}\is {\cal T}({m})+\varphi^{(o)}\is N^{\mathbb G}+\varphi^{(pr)}\is D.\\
 &&{\varphi^{(pr)}}> -\big[ {G_{-}}\psi_2 + \varphi^{(o)} {G}\big]/{\tilde{G}},\ \mbox{and}\    -\frac{\psi_2 G_-} {G}<\varphi^{(o)}<   \frac{\psi_2G_-} {\Delta D^{o,\mathbb F}}\quad P\otimes dD\mbox{-a.e.},\hskip 1cm\\
 &&\mbox{and}\quad \mu+  \psi_1\sigma + \psi_2 \zeta\lambda\equiv 0, \quad \psi_2>-1\quad P\otimes dt-a.e.
\end{eqnarray*}
{\rm{(c)}}  There exists a unique quadruplet $(\psi_1,\psi_2, \varphi^{(o)}, \varphi^{(pr)})$ that belongs to the set $L^1_{loc}(W,\mathbb F)\times L^1_{loc}(N^{\mathbb F},\mathbb F)\times{\cal I}^o_{loc}(N^{\mathbb G},\mathbb G)\times L^1_{loc}( \rm{Prog}(\mathbb F),P\otimes dD)$, and satisfies (\ref{Condition4Fi(pr)}), 
\begin{eqnarray*}\label{repZG1a}
 && Z^{\mathbb G}={{{ \cal E}(\psi_1\is W+\psi_2\is N^{\mathbb F})^{\tau}}\over{{\cal E}(G_{-}^{-1}\is m)^{\tau}}}{\cal E}(\varphi^{(o)}\is N^{\mathbb G}){\cal E}(\varphi^{(pr)}\is D),\\
 &&{\varphi^{(pr)}}> -1,\quad P\otimes dD\mbox{-a.e.},\quad   - \frac{ G_-} {G}<\varphi^{(o)}<   \frac{G_-} {\Delta D^{o,\mathbb F}}\quad P\otimes dD^{o,\mathbb F}\mbox{-a.e.},\hskip 1cm\\
 &&\mbox{and}\quad\mu+  \psi_1\sigma + \psi_2\zeta\lambda\equiv 0, \ \ \psi_2 > -1 \ \ P\otimes dt-a.e.
\end{eqnarray*}
\end{theorem}
\begin{proof}
The proof follows immediately  from Theorems \ref{LocalMartingaleDeflator} and the fact that for any $M\in {\cal M}_{loc}(\mathbb F)$, there exists a unique pair of $\mathbb F$-predictable processes $(\psi_1,\psi_2)\in  L^1_{loc}(W,\mathbb F)\times L^1_{loc}(N^{\mathbb F},{\mathbb F})$ such that $M=M_0+\psi_1\is W+\psi_2\is N^{\mathbb F}.$\end{proof}
\subsubsection{The discrete-time market models}
In this subsection, we suppose that the trading times are $t=0,1,...,T$, and hence  $\tau$ takes values in $\{0,1,..., T\}$ and on $(\Omega, {\cal F},P)$ we have 
\begin{eqnarray}\label{GandG*}
 \mathbb{F}:= ({\cal F}_n)_{n=0,1,..., T},\ {\cal G}_n={\cal F}_n\vee \sigma\left(\tau\wedge 1,...,\tau\wedge n\right),\ S=(S_n)_{n=0,1,..., T},\end{eqnarray}
$$G_n=\sum_{k=n+1}^T P(\tau =k | {\cal F}_n),\  \widetilde{G}_n=\sum_{k=n}^T P(\tau =k | {\cal F}_n),\ n=0,...,T.$$

Then the discrete-time version of the operator ${\cal T}$  and the $\mathbb G$-martingale $N^{\mathbb G}$ defined in (\ref{processMhat})  and (\ref{processNG}) respectively are given by 
\begin{eqnarray}
{\cal T}_n(M)=\sum_{k=1}^{n\wedge\tau} {{P(\tau\geq k|{\cal F}_{k-1})}\over{P(\tau\geq k|{\cal F}_{k})}}\Delta M_k+\sum_{k=1}^{n\wedge\tau} E[\Delta M_k I_{\{P(\tau\geq k|{\cal F}_{k})=0\}}|{\cal F}_{k-1}],\label{OpertaorTdiscrete}\end{eqnarray}
where $\Delta M_n:=M_n-M_{n-1}$ and $M$ is an $\mathbb F$-martingale, and
\begin{eqnarray}
N^{\mathbb G}_n:=I_{\{\tau\leq n \}}-\sum_{k=1}^{n\wedge\tau}{{P(\tau=k|{\cal F}_{k-1})}\over{P(\tau\geq k|{\cal F}_k)}},\quad  n=1,...,T.\label{ProcessNGdiscrete}
\end{eqnarray}
Then $\tau$ is a $\mathbb G$-stopping time and ${\cal G}_{\tau}$ is defined as usual, while ${\cal F}_{\tau}$ is given by ${\cal F}_{\tau}:=\sigma(\{X_{\tau},\ X\ \mathbb F\mbox{-adapted and bounded}\})$. Below, we discuss the relationship between ${\cal G}_{\tau}$ and ${\cal F}_{\tau}$, as  this role is very important in our analysis. 

\begin{lemma}\label{LemmaGnFn}Consider the discrete-time setting of (\ref{GandG*}). Then the $\sigma$-fields ${\cal F}_{\tau}$ and ${\cal G}_{\tau}$ coincide, and hence for any ${\cal G}_{\tau}$-measurable random variable $X$, there exists an $\mathbb F$-adapted process $\xi$ such that $X=\xi_{\tau}$ $P$-a.s.
\end{lemma}
\begin{proof}
Since $\tau$ is a $\mathbb G$-stopping time and due to \cite[Theorem 64, Chapter IV]{Dellacherie72}, we conclude that for any ${\cal G}_{\tau}$-measurable random variable $X$, there exists a $\mathbb G$-adapted process, $\xi^{\mathbb G}=(\xi^{\mathbb G}_n)_{n=0,1,...,T}$ such that $X=\xi^{\mathbb G}_{\tau}$. Thus, the lemma follows immediately if we prove that for any $n\in\{0,...,T\}$, and any ${\cal G}_n$-measurable random variable $X^{\mathbb G}_n$, there exists an ${\cal F}_n$-measurable random variable $ X^{\mathbb F}_n$ such that 
\begin{eqnarray}\label{equalityGnFn}
X^{\mathbb G}_n=X^{\mathbb F}_n\quad\mbox{on}\quad (\tau=n).\end{eqnarray}
Thus, on the one hand, it is clear that (\ref{equalityGnFn}) holds for random variables having the form of $X^{\mathbb G}_n=\xi^{\mathbb F}_n f(\tau\wedge1, ...,\tau\wedge n)$, where $\xi^{\mathbb F}_n$ is a bounded and ${\cal F}_n$-measurable random variable and $f$ is a bounded and Borel-measurable real-valued function on $\mathbb R^n$. On the other hand, these random variables generate the vector space of bounded  and ${\cal G}_n$-measurable random variables. Hence the fulfillment of (\ref{equalityGnFn}) for general random variables, follows from this remark and the class monotone theorem (see \cite[Theorem 21, Chapter I]{Dellacherie72}). This proves the lemma.\end{proof}
The impact of Lemma \ref{LemmaGnFn} can be noticed immediately in the discrete-time version of Theorem \ref{theo4MartingaleDecomposGeneral}, that we state below.
\begin{theorem} Let $M^{\mathbb G}$ be a $\mathbb G$-local martingale. Then there exist an $\mathbb F$-local martingale, $M^{\mathbb F}$, and an $\mathbb F$-adapted processes  $\varphi$ such that  
\begin{eqnarray}\label{ChoulliDecompoDiscrete}
M^{\mathbb G}_{n\wedge \tau}=M^{\mathbb G}_0+\sum_{k=1}^{n}{{\Delta{\cal T}_k(M^{\mathbb F})}\over{P(\tau\geq k|{\cal F}_{k-1})^2}}+\sum_{k=1}^n \varphi_k\Delta N^{\mathbb G}_k.\end{eqnarray}
\end{theorem}
\begin{proof} The proof follows from combining Theorem \ref{theo4MartingaleDecomposGeneral}  and Lemma \ref{LemmaGnFn}.
\end{proof}
Below, we state our main result in this subsection.

\begin{theorem}\label{DiscreteTimeDeflators}  Let $Z^{\mathbb G}$ be a positive and $\mathbb G$-adapted process, and ${\widetilde Q}$ be the probability measure on $\mathbb F$ given by 
 \begin{equation}\label{Qprobability}
  {\widetilde Q}:={\overline Z}_T\cdot P,\ \mbox{where}\ \ {\overline Z}_n:=\prod_{k=1}^n \left({{\widetilde G_k}\over{G_{k-1}}}I_{\{G_{k-1}>0\}}+I_{\{G_{k-1}=0\}}\right).
 \end{equation}
 Then the following assertions are equivalent.\\
 {\rm{(a)}}  $Z^{\mathbb G}$ is a deflator for $(S^{\tau}, \mathbb G)$  (i.e., $Z^{\mathbb G}\in {\cal D}(S^{\tau}, \mathbb G)$).\\
 {\rm{(b)}}   There exists a unique pair $\left(Z^{({\widetilde Q},\mathbb F)}, \varphi\right)$ such that $Z^{({\widetilde Q},\mathbb F)}\in{\cal D}(S, {\widetilde Q},\mathbb F)$, $\varphi$ is an $\mathbb F$-adapted process satisfying for all $n=0,...,T$ $P\mbox{-a.s.}$
  \begin{eqnarray}\label{ineqDiscreteTime} 
-{{P(\tau\geq n|{\cal F}_n)}\over{ P(\tau> n|{\cal F}_n)}}<\varphi_n<{{P(\tau\geq n|{\cal F}_n)}\over{P(\tau= n|{\cal F}_n)}},\quad \mbox{and}\quad Z^{\mathbb G}=(Z^{({\widetilde Q},\mathbb F)})^{\tau}Z^{(\varphi)}. \end{eqnarray}
Here $Z^{(\varphi)}$ is given by 
\begin{eqnarray}\label{repKGDiscreteTime}
\ Z^{(\varphi)}_t:=\prod_{n=1}^t \left(1+\varphi_n{{P(\tau>n|{\cal F}_n)}\over{P(\tau\geq n|{\cal F}_n)}}I_{\{\tau=n\}}-\varphi_n{{P(\tau=n|{\cal F}_n)}\over{P(\tau\geq n|{\cal F}_n)}}I_{\{\tau> n\}}\right).\end{eqnarray}
\end{theorem}
\begin{proof} We start this proof by making the following three remarks:\\
1) It is easy to check that (see also \cite{CD1} for details and related results) the process $\overline Z$ is an $\mathbb F$-martingale and hence $\widetilde Q$ is a well defined probability measure. Furthermore, the process $\overline Z$ is the discrete-time version of the process ${\cal E}(G_{-}^{-1}I_{\{G_{-}>0\}}\is m)$ (which is well defined even in the case where $G$ might vanish) see \cite[Subsection 2.3]{Kardaras2010}.\\
2) It is clear that $X$ is a $\widetilde Q$-supermartingale iff $Y:={\overline Z}X$ is a supermartingale.\\
3) Thanks to (\ref{ProcessNGdiscrete}), the discrete-time version of ${\cal E}(\varphi\is N^{\mathbb G})$ coincides with $Z^{(\varphi)}$ given in (\ref{repKGDiscreteTime}), for any $\mathbb F$-optional process $\varphi$. \\
Then by combining these three remarks and Theorem \ref{GeneralDeflators}, the proof of the theorem follows immediately.
\end{proof}
\section{NFLVR  for market models stopped at ${\tau}$} \label{Section4NFLVR}
  This section discusses an important application of Section \ref{SectionOfDeflators} to  the stability of {\bf no free lunch with vanishing risk}  concept (NFLVR hereafter) and its variant under stopping with $\tau$. To this end, we recall that $(L^{\infty}(P), \Vert\cdot\Vert_{\infty})$ is the set of bounded random variables endowed with its norm, we put 
  \begin{eqnarray}\label{Linfinity(X)}
  L_{\infty}(X, \mathbb H):=\left\{ H\in L(X, \mathbb H):\  \lim_{t\longrightarrow+\infty} (H\is X)_t\ \mbox{exists almost surely}\right\},\end{eqnarray}
and we recall the mathematical definition of NFLVR and related concepts.
  \begin{definition}\label{NFLVRDefinition} Let $(X, \mathbb H)$ be a market model, where $X$ is the discounted assets' prices, and $\mathbb H$ is the filtration. \\
  {\rm{(a)}} The model  $(X, \mathbb H)$ is said to satisfy NFLVR if there is no sequence of pair $(H_n, \alpha_n)\in L_{\infty}(X, \mathbb H)\times[0,+\infty)$ satisfying $H_n\is X\geq -\alpha_n$, $(H_n\is X)_{\infty}$ converges in probability to a random variable, $f$, with values in $[0,+\infty]$ and $P(f>0)>0$, and $\Vert\max(0,- (H_n\is X)_{\infty})\Vert_{\infty}$ converges to zero.\\
   {\rm{(b)}} A martingale density for the model  $(X, \mathbb H)$ is any local martingale deflator $Z$ for this model  (i.e., $Z\in {\cal Z}_{loc}(X, \mathbb H)$) that is a uniformly integrable martingale. The set of all martingale densities is denoted by $ {\cal Z}(X, \mathbb H)$.
      \end{definition}
 It is clear that for any $T\in (0,+\infty)$ and $\mathbb H$-semimartinagle $X$, we have $ L_{\infty}(X^T, \mathbb H)= L(X^T, \mathbb H)$, and the limit at infinity in this case is irrelevant. For the first definition of NFLVR, we refer the reader to \cite{DelbaenSchachermayer94,DelbaenSchachermayer99} and the references therein. In fact, our Definition \ref{NFLVRDefinition}-(a) is exactly \cite[Corollary 3.7]{DelbaenSchachermayer94}.  On the one hand, thanks to  \cite[Theorem 1.1]{DelbaenSchachermayer99}, it is known that $(X, \mathbb H)$ satisfies NFLVR iff it admits a martingale density (i.e.,  ${\cal Z}(X, \mathbb H)\not=\emptyset$) . On the other hand, in virtue of Theorem \ref{LocalMartingaleDeflator}, all martingale densities --when they exist-- take the form of (\ref{repKG1}) or equivalently  the multiplicative form of (\ref{repKGMulti}). Therefore, the NFLVR concept for $(S^{\tau},\mathbb G)$ boils down to single out whether there exist uniformly integrable martingales among the local martingale deflators.\\
 
The following theorem states our results on NFLVR for the models $(S^{\tau\wedge T}, \mathbb G)$ and $(S^{\tau}, \mathbb G)$, where  $T\in (0,+\infty)$ is a fixed horizon. To this end, through the rest of the paper, as $P(\tau=+\infty)>0$ in general, for any process $X$ having the limit at infinity $P$-almost surely, we put 
$$X_{\tau}=X_{\infty}:=\displaystyle\lim_{t\longrightarrow+\infty} X_t\quad \mbox{on}\quad (\tau=+\infty).$$
  \begin{theorem}\label{NFLVR}
  Suppose that $G>0$. Then the following assertions hold.\\
  {\rm{(a)}} For any fixed horizon $T\in (0,+\infty)$,  the model $(S^{\tau\wedge T},\mathbb G)$ satisfies NFLVR as soon as the model $(S^T, \mathbb F)$ does. Furthermore, for any  $Z^{\mathbb F}\in{\cal Z}(S^T, \mathbb F)$,  we have $(Z^{\mathbb F})^{T\wedge\tau}/{\cal E}(G_{-}^{-1}\is m)^{T\wedge\tau}\in {\cal Z}(S^{\tau\wedge T}, \mathbb G)$.\\
   {\rm{(b)}} For any bounded $(\varphi^{(o)},\varphi^{(pr)})\in {\cal I}^o_{loc}(N^{\mathbb G},\mathbb G)\times  L^1_{loc}({\rm{Prog}}(\mathbb F),P\otimes dD)$ such that $E[\varphi^{(pr)}_{\tau}\big|{\cal F}_{\tau}]I_{\{\tau<+\infty\}}=0$ $P$-a.s. and 
  \begin{eqnarray}\varphi^{(pr)}>-1,\quad 0\leq \varphi^{(o)},\quad  \varphi^{(o)}(\widetilde G-G)<\widetilde G,\quad  P\otimes dD\mbox{-a.e.},\label{Condition4pairPhi}\end{eqnarray}
  and any $Z\in {\cal Z}(S^T, \mathbb F)$, the process 
  $${{Z^{\tau\wedge T}}\over{{\cal E}\left(G_{-}^{-1}\is m\right)^{\tau\wedge T}}}{\cal E}(\varphi^{(o)}\is N^{\mathbb G})^T{\cal E}(\varphi^{(pr)}\is D)^T\quad\mbox{belongs to}\quad {\cal Z}(S^{\tau\wedge T},\mathbb G).$$
   {\rm{(c)}} If 
   \begin{eqnarray}\label{NSC4NFLVR}
   \biggl\{ {\cal E}_{\infty}\left(-{G_{-}}^{-1}\is m\right)=0\biggr\}\subset\biggl\{{\cal E}_{\infty}\left(-{\widetilde G}^{-1}\is D^{o,\mathbb F}\right)=0\biggr\},\quad P\mbox{-a.s.,}\end{eqnarray}
   then $(X^{\tau}, \mathbb G)$ satisfies NFLVR, for any model $(X, \mathbb F)$ satisfying NFLVR. Furthermore, under (\ref{NSC4NFLVR}), for any $Z\in  {\cal Z}(X, \mathbb F)$ and any pair $(\varphi^{(0)}, \varphi^{(pr)})$ fulfilling the conditions of  assertion (b), we have 
   \begin{eqnarray}\label{SetZ(S, tau)}
 {{Z^{\tau}{\cal E}(\varphi^{(o)}\is N^{\mathbb G})}\over{{\cal E}\left(G_{-}^{-1}\is m\right)^{\tau}}}{\cal E}(\varphi^{(pr)}\is D)\quad\mbox{belongs to}\quad {\cal Z}(X^{\tau}, \mathbb G).\end{eqnarray}
  {\rm{(d)}}  Suppose that 
  \begin{eqnarray}\label{NSC4NFLVRbis}
{\cal E}_{\infty}\left(-{\widetilde G}^{-1}\is D^{o,\mathbb F}\right)=0,\quad P\mbox{-a.s.}.\end{eqnarray}
Then  $(X^{\tau}, \mathbb G)$ satisfies NFLVR, for any model $(X, \mathbb F)$ satisfying NFLVR. 
    \end{theorem}
Thanks to Lemma \ref{Gdecomposition}-(c), all quantities in Theorem \ref{NFLVR} are well defined. The rest of this section is divided into three subsections. The first subsection gives a detailed proof for Theorem \ref{NFLVR}, and outlines some intermediate results.  The second subsection discusses the novelties of Theorem \ref{NFLVR} and/or shows how it is consistent with the existing literature. The last subsection deals with the stability of a variant of NFLVR under stopping with $\tau$.
\subsection{Proof of Theorem \ref{NFLVR} and other related results}
In this subsection, we address the proof of Theorem \ref{NFLVR}. To this end, we start by remarking that assertion (a) follows directly from assertion (b) by putting $\varphi^{(o)}=\varphi^{(pr)}\equiv 0$.  Assertion (d) follows immediately from assertion (c), as in general, (\ref{NSC4NFLVRbis}) clearly implies (\ref{NSC4NFLVR}). Furthermore, when $\tau<+\infty$ $P$-a.s., the two conditions (\ref{NSC4NFLVR}) and (\ref{NSC4NFLVRbis}) are equivalent, and this fact is a direct consequence of Lemma \ref{Gdecomposition}-(e).
Thus, the remaining part of this subsection deals with the proofs of assertions (b) and (c) of  Theorem \ref{NFLVR}.  Surprisingly, we  could not deduce assertion (b) from assertion (c) by letting $T$ to be infinite, and both assertions require separate proofs. However, an important part of these proofs relies essentially on the following proposition that is interesting in itself. 
    \begin{proposition}\label{Proposition4NFLVR} The following assertions hold.\\
      {\rm{(a)}} For any $N\in {\cal M}(\mathbb F)$, the process $N^{\tau}/{\cal E}(G_{-}^{-1}\is m)^{\tau}$ belongs to ${\cal M}(\mathbb G)$. In particular $1/{\cal E}(G_{-}^{-1}\is m)^{\tau}$ is a $\mathbb G$-martingale and for any $T\in (0,+\infty)$,
     \begin{eqnarray}\label{Qtilde}
  \widetilde Q_T:={\widetilde Z}_T\cdot P\quad\rm{where}\quad  {\widetilde Z}:={\cal E}\left(-G_{-}^{-1}\is {\cal T}(m)\right)=1/{\cal E}(G_{-}^{-1}\is m)^{\tau}
  \end{eqnarray}
    is a well defined probability measure on $(\Omega, {\cal G}_T)$.\\
      {\rm{(b)}} Let $N$be a uniformly integrable $\mathbb F$-martingale with $N_{\infty}>0$ $P$-a.s.. Then $N/{\cal E}(G_{-}^{-1}\is m)^{\tau}$ is a uniformly integrable $\mathbb G$-martingale if and only if (\ref{NSC4NFLVR}) holds, or equivalently $P$-a.s. we have 
      \begin{eqnarray}\label{NSC4UniformIntegrability} 
      \left\{ {\cal E}_{\infty}(-{1\over{G_{-}}}\is m)=0\right\}\subset\left\{ \int_0^{\infty} {{I_{\{\widetilde G_s=G_s\}}}\over{G_s}} dD^{o,\mathbb F}_s+\sum_{s>0}\ln({{{\widetilde G}_s}\over{G_s}})=\infty\right\}.
      \end{eqnarray}
      In this case, $\widetilde Q:={\widetilde Z}_{\infty}\cdot P$ is a well defined probability measure on $(\Omega, {\cal F})$.
     \end{proposition}
  The proof of the proposition is given jointly with the proof of Theorem \ref{NFLVR}, and are based on the following key lemma. 
  \begin{lemma}\label{keylemma}
  If $N$ is a uniformly integrable $\mathbb F$-martingale, then the nondecreasing process $-\vert N\vert \is {\cal E}(-{\widetilde G}^{-1}\is D^{o,\mathbb F})$ is integrable (i.e., $$E\left[\int_0^{\infty} -\vert N_s\vert d{\cal E}_s\left(-{\widetilde G}^{-1}\is D^{o,\mathbb F}\right)\right]<\infty)$$ and ${\cal E}_{-}(-{\widetilde G}^{-1}\is D^{o,\mathbb F})\is N$ is a uniformly integrable $\mathbb F$-martingale.
  \end{lemma}
  For the sake of simplicity, the proof of this lemma is relegated to Appendix \ref{SectionC}. Below, we prove Theorem \ref{NFLVR}  and Proposition \ref{Proposition4NFLVR}.
 \begin{proof}{\it of Theorem \ref{NFLVR} and Proposition \ref{Proposition4NFLVR}:} This proof is divided into three parts. The first part proves Proposition \ref{Proposition4NFLVR}-(a) and Theorem \ref{NFLVR}-(a), while Theorem \ref{NFLVR}-(b) is proved in the second part. The third part addresses Proposition \ref{Proposition4NFLVR}-(b) and Theorem \ref{NFLVR}-(c).\\
 {\bf Part 1.} Remark that  the proof of Theorem \ref{NFLVR}-(a) follows from combining Proposition  \ref{Proposition4NFLVR}-(a) and Corollary \ref{caseofGmartingaleDensities}. Furthermore, the last statement of Proposition  \ref{Proposition4NFLVR}-(a) follows immediately from the first statement of this assertion by considering $N\equiv 1$. Thus, the rest of this part focuses on proving the first statement of Proposition \ref{Proposition4NFLVR}-(a). To this end, we consider $N\in {\cal M}(\mathbb F, P)$, and for any bounded $\mathbb F$-stopping time $\sigma$ we derive 
  \begin{eqnarray}
    &&E\left[{{ N_{\sigma\wedge\tau}}\over{{\cal E}_{\sigma\wedge\tau}\left(G_{-}^{-1}\is m\right)}}\right]\nonumber\\
    &&= E\left[{{ N_{\sigma}}\over{{\cal E}_{\sigma}\left(G_{-}^{-1}\is m\right)}}I_{\{\tau>{\sigma}\}}\right]+ E\left[{{ N_{\tau}}\over{{\cal E}_{\tau}\left(G_{-}^{-1}\is m\right)}}I_{\{0<\tau\leq {\sigma}\}}+N_0 I_{\{\tau=0\}}\right]\nonumber\\
  &&= E\left[{{ N_{\sigma}}\over{{\cal E}_{\sigma}\left(G_{-}^{-1}\is m\right)}}G_{\sigma}+N_0(1-G_0)\right]+ E\left[\int_0^{\sigma}{{ N_{s}}\over{{\cal E}_{s}\left(G_{-}^{-1}\is m\right)}}\delta_{\tau}(ds)\right]\nonumber\\
  &&=E\left[N_0+G_0\left(N_{\sigma}{\cal E}_{\sigma}\left(- {1\over{\widetilde G}}\is  D^{o,\mathbb F}\right)-N_0\right)+ \int_0^{\sigma}{{ N_s}\over{{\cal E}_s\left(G_{-}^{-1}\is m\right)}} dD^{o,\mathbb F}_s\right]\nonumber\\
   &&=E[N_0] + E\left[G_0\int_0^{\sigma}{\cal E}_{s-}\left(-  {\widetilde G}^{-1}\is  D^{o,\mathbb F}\right)dN_s\right]\nonumber\\
   &&-E\left[\int_0^{\sigma} {{G_0}\over{ {\widetilde G}_s}}N_s{\cal E}_{s-}\left(- {1\over{\widetilde G}}\is  D^{o,\mathbb F}\right)dD^{o,\mathbb F}_s\right]+ E\left[\int_0^{\sigma} {{ N_s}\over{{\cal E}_s\left(G_{-}^{-1}\is m\right)}}\is D^{o,\mathbb F}_s\right]\nonumber\\
   &&= E[N_0] + E\left[\int_0^{\sigma}G_0{\cal E}_{s-}\left(- {\widetilde G}^{-1}\is  D^{o,\mathbb F}\right)dN_s\right].\label{ConstantExpectation}
    \end{eqnarray}
 The third and the last equalities, in the above equalities, follow from the following facts, which are consequences of (\ref{MultiDecompo4G}) and $\widetilde G=G+\Delta D^{o,\mathbb F}$, 
     \begin{eqnarray}\label{Grelationship2m}
\left({\cal E}\left({G_{-}}^{-1}\is m\right)\right)^{-1}={{G_0}\over{G}}{\cal E}\left(-  {1\over{\widetilde G}}\is  D^{o,\mathbb F}\right)={{G_0}\over{\widetilde G}}{\cal E}_{-}\left(-  {1\over{\widetilde G}}\is  D^{o,\mathbb F}\right).
  \end{eqnarray} 
  For any $T\in (0,+\infty)$ and any $\mathbb F$-martingale $N$, by applying Lemma \ref{keylemma} to $N^T$, we deduce that $G_0 {\cal E}_{-}\left(- {\widetilde G}^{-1}\is  D^{o,\mathbb F}\right)\is N^T$ is a uniformly integrable $\mathbb F$-martingale. Thus, by combining this latter fact with (\ref{ConstantExpectation}),  we obtain
\begin{eqnarray*}
E\left[{{ N_{\sigma\wedge\tau}}\over{{\cal E}_{\sigma\wedge\tau}\left(G_{-}^{-1}\is m\right)}}\right]= E[N_0]\ \mbox{for any bounded $\mathbb F$-stopping time $\sigma$}.\end{eqnarray*}
A combination of this with Lemma \ref{PortfolioGtoF}-(d) implies that $N^{\tau}/{\cal E}(G_{-}^{-1}\is m)^{\tau}$ is a $\mathbb G$-martingale. This proves Proposition \ref{Proposition4NFLVR}-(a), and ends part 1.\\
 {\bf Part 2.} To prove Theorem \ref{NFLVR}-(b), we consider $Z\in {\cal Z}(S^T, \mathbb F)$ and a pair of bounded processes $(\varphi^{(o)}, \varphi^{(pr)})\in {\cal I}^o_{loc}(N^{\mathbb G},\mathbb G)\times  L^1_{loc}({\rm{Prog}}(\mathbb F),P\otimes dD)$ such that  $E[\varphi^{(pr)}_{\tau}\big|{\cal F}_{\tau}]I_{\{\tau<+\infty\}}=0$ $P$-a.s.  and (\ref{Condition4pairPhi}) holds. Therefore, due to $\ln(1+x)-x\leq 0$ for any $x>-1$, we get  
\begin{eqnarray}
 0\leq{\cal E}_t(\varphi^{(pr)}\is D)&&=(1+\varphi^{(pr)}_{\tau}I_{\{\tau\leq t\}})\leq 1+ C\label{Bound1}\\
 0\leq {\cal E}_t(\varphi^{(o)}\is N^{\mathbb G})&&={\cal E}_t\left(-{{\varphi^{(o)}}\over{\widetilde G}}I_{\Lbrack0,\tau\Rbrack}\is D^{o,\mathbb F}\right)\left(1+{{G_{\tau}}\over{\widetilde G_{\tau}}}\varphi^{(o)}_{\tau}I_{\{\tau\leq t\}}\right)\nonumber\\
 &&\leq 1+C.\hskip 0.75cm \label{Bound2}
 \end{eqnarray} 
Thus, by using these inequalities and assertion (a), we conclude that the process ${\cal E}(\varphi^{(o)}\is N^{\mathbb G}){\cal E}(\varphi^{(pr)}\is D)Z^{\tau}/{\cal E}\left(G_{-}^{-1}\is m\right)^{\tau}$, which belongs to ${\cal Z}_{loc}(S^{\tau\wedge T}, \mathbb G)$, is a $\mathbb G$-martingale. This proves assertion (b) of the theorem.\\
{\bf Part 3.}  Herein, we prove both Theorem \ref{NFLVR}-(c) and Proposition \ref{Proposition4NFLVR}-(b) as they are related somehow. In fact, due to (\ref{DDequation}),  it is easy to see that the conditions (\ref{NSC4NFLVR}) and (\ref{NSC4UniformIntegrability}) are equivalent. Suppose that (\ref{NSC4NFLVR})  holds, and consider a model $(X, \mathbb F)$ satisfying the NFLVR. Then, we deduce the existence of  $Z\in {\cal Z}(X, \mathbb F)$ and $Z^{\tau}/{\cal E}(G_{-}^{-1}\is m)^{\tau}\in {\cal Z}_{loc}(X^{\tau},\mathbb G)$ due to Corollary \ref{caseofGlocalmartingale}. As $Z$ is a uniformly integrable $\mathbb F$-martingale  with $Z_{\infty}>0$ $P$-a.s., by applying Proposition \ref{Proposition4NFLVR}-(b) to $Z$, we conclude that $Z^{\tau}/{\cal E}(G_{-}^{-1}\is m)^{\tau}\in {\cal Z}(X^{\tau},\mathbb G)$. This proves the first statement of Theorem \ref{NFLVR}-(c), while the second statement follows directly from combining the first statement and the arguments in part 2, mainly (\ref{Bound1})-(\ref{Bound2}). This proves Theorem \ref{NFLVR}-(c). Thus, the rest of this part addresses Proposition \ref{Proposition4NFLVR}-(b). To this end, we consider a uniformly integrable $\mathbb F$-martingale $N$ with $N_{\infty}>0$ $P$-a.s.. Thus, due to It\^o  applied to $N{\cal E}\left(- {\widetilde G}^{-1}\is  D^{o,\mathbb F}\right)$, (\ref{Ginfinity}) and Lemma \ref{Gdecomposition}-(d), we put $\Gamma_+:=\{{\cal E}_{\infty}(G_{-}^{-1}\is m)>0\}$ and derive 
  \begin{eqnarray}
   &&E\left[{{ N_{\tau}}\over{{\cal E}_{\tau}\left(G_{-}^{-1}\is m\right)}}\right]\nonumber\\
   &&=  E\left[{{ N_{\infty}}\over{{\cal E}_{\infty}\left(G_{-}^{-1}\is m\right)}}I_{\{\tau=\infty\}}+{{ N_{\tau}}\over{{\cal E}_{\tau}\left(G_{-}^{-1}\is m\right)}}I_{\{0<\tau<+\infty\}}+N_0 I_{\{\tau=0\}}\right]\nonumber\\
  &&= E\left[N_0(1-G_0)+G_0 N_{\infty}{\cal E}_{\infty}(-{1\over{\widetilde G}}\is D^{o,\mathbb F})I_{\Gamma^+}+\int_0^{\infty}{{ N_{s}}\over{{\cal E}_{s}\left(G_{-}^{-1}\is m\right)}}\delta_{\tau}(ds)\right]\nonumber\\
  &&= E\left[N_0(1-G_0)+G_0 N_{\infty}{\cal E}_{\infty}(-{1\over{\widetilde G}}\is D^{o,\mathbb F})I_{\Gamma^+}+\int_0^{\infty}{{ N_{s}}\over{{\cal E}_{s}\left(G_{-}^{-1}\is m\right)}}dD^{o,\mathbb F}_s\right]\nonumber\\
   &&=  E\left[N_0(1-G_0)+G_0 N_{\infty}{\cal E}_{\infty}(-{1\over{\widetilde G}}\is D^{o,\mathbb F})I_{\Gamma^+}+G_0\int_0^{\infty} {{N_s }\over{\widetilde G_s}}{\cal E}_{s-}(- {1\over{\widetilde G}}\is  D^{o,\mathbb F})dD^{o,\mathbb F}_s\right]\nonumber\\
   &&= E\left[N_0(1-G_0)+G_0 N_{\infty}{\cal E}_{\infty}(-{1\over{\widetilde G}}\is D^{o,\mathbb F})I_{\Gamma^+}-G_0\int_0^{\infty} N_s d{\cal E}_{s}\left(- {\widetilde G}^{-1}\is  D^{o,\mathbb F}\right)\right]\nonumber\\
  &&=E\left[N_0- G_0N_{\infty}{\cal E}_{\infty}(-{1\over{\widetilde G}}\is D^{o,\mathbb F})I_{\Omega\setminus\Gamma^+}+ \int_0^{\infty} G_0{\cal E}_{s-}\left(-{\widetilde G}^{-1}\is D^{o,\mathbb F}\right)dN_s\right]\nonumber\\
   &&=E\left[N_0- G_0N_{\infty}{\cal E}_{\infty}\left(- {\widetilde G}^{-1}\is  D^{o,\mathbb F}\right)I_{\Omega\setminus\Gamma^+}\right].\label{martingale4tau}
   \end{eqnarray}
   The last equality is a direct consequence of Lemma \ref{keylemma}, while the fourth equality follows from (\ref{Grelationship2m}). Therefore, in virtue of (\ref{martingale4tau}), we conclude that the positive process $ N^{\tau}/{\cal E}\left(G_{-}^{-1}\is m\right)^{\tau}$ is a uniformly integrable $\mathbb G$-martingale if and only if $G_0N_{\infty}{\cal E}_{\infty}\left(- {\widetilde G}^{-1}\is  D^{o,\mathbb F}\right)I_{\Omega\setminus\Gamma^+}=0$ $P$-a.s. or equivalently (\ref{NSC4NFLVR}) holds (due to $G_0>0$ and $N_{\infty}>0$ $P$-a.s.).  This ends the proof of Theorem \ref{NFLVR} and Proposition \ref{Proposition4NFLVR}.
    \end{proof}
  \subsection{ What are the novelties of Theorem \ref{NFLVR} compared to the literature?}
  As we mentioned in the introduction, the three papers \cite{ACDJ0, CJN} and \cite{Fontana} sound to be the most advanced literature related to NFLVR under stopping with $\tau$. In \cite[Corollary 4.6]{CJN} and \cite[Lemma 3.1]{ACDJ0} (respectively in \cite[Lemma 3.3]{ACDJ0}), the authors prove that NFLVR holds for $(S^{\tau}, \mathbb G)$ as soon as the immersion assumption (respectively  the positive density hypothesis) is fulfilled. It is important to mention that these two conditions on the pair $(\tau, \mathbb F)$ are very restrictive, as both boil down to a sort of ``independence" between $\tau$ and $\mathbb F$ somehow, see \cite[Remark 4.3]{ACDJ0}. In the case of immersion, which is a a stronger case than the case when $\tau$ is a pseudo-stopping introduced in \cite{Nik2005}, we have $m\equiv m_0$ and hence (\ref{NSC4NFLVR}) holds as ${\cal E}(G^{-1}_{-}\is m)\equiv 1$. Furthermore, in this case, there is equivalence between $\tau<\infty$ $P$-a.s. and ${\cal E}_{\infty}(-{\widetilde G}^{-1}\is D^{o,\mathbb F})=0$ $P$-a.s., while in general the latter implies the former only. In contrast to \cite{CJN}, the immersion assumption is relaxed in \cite{Fontana} while the  model $(S, \mathbb F, \tau)$ satisfies the following several assumptions instead:
  \begin{eqnarray}
  &&\hskip -1.5cm S\ \mbox{is {\bf continuous} and fulfills NFLVR (and $M$ is its martingale part)},\label{Fontana1}\\
  &&\hskip -1.5cm\tau\ \mbox{avoids $\mathbb F$-stopping times,}\ i.e., P(\tau=\sigma)=0,\ \forall\ \sigma\ \mbox{finite $\mathbb F$-stopping,}\label{Fontana2}\\
  &&\hskip -1.5cm (M,\mathbb F) \mbox{ has the predictable representation property,}\label{Fontana3}\\
  &&\hskip -1.5cm \tau\ \mbox{is a honest time}, i.e., {\widetilde G}_{\tau}=1\ P\mbox{-a.s. on}\ (\tau<+\infty).\label{Fontana4}
  \end{eqnarray}
  Under these assumptions, that are vital for their analysis, the authors state in \cite[Theorem 5.1]{Fontana} that NFLVR holds for $(S^{\tau\wedge\sigma}, \mathbb G)$ iff $P(\sigma\geq R_0)=0$, where $R_0$ is defined in (\ref{Rzero}). Theorem \ref{NFLVR}-(a) is valid under $G>0$ (or equivalently $R_0=+\infty$ $P$-a.s.), and hence $P(\sigma\geq R_0)=0$ for any bounded $\mathbb F$-stopping time  $\sigma$. This shows that Theorem \ref{NFLVR}-(a) is consistent with \cite[Theorem 5.1]{Fontana}. Furthermore, in \cite[Theorem 4.1]{Fontana}, it is shown that NFLVR  fails for $(S^{\tau}, \mathbb G)$ when (\ref{Fontana1})-(\ref{Fontana2})-(\ref{Fontana3})-(\ref{Fontana4}) are enforced. Hence, up to our knowledge, there is no result on the stability of NFLVR under stopping with a possibly unbounded random horizon, except for the immersion and the positive density hypothesis cases. In this spirit, Theorem \ref{NFLVR}-(c) answers partially, but in a more general setting, the open question of whether there are random times that do not alter NFLVR after stopping?\\
  
 To explain how consistent Theorem \ref{NFLVR} with \cite[Theorem 4.1]{Fontana}, \cite{ACDJ0} and the references therein, we state the following lemma. It presents two classes of models of $\tau$ that violate assumptions of Theorem \ref{NFLVR}-(c) and (d). 
\begin{lemma}\label{LemmaNickYor} Suppose that there exists a positive and continuous $\mathbb F$-martingale $N$ (i.e., $N_t>0$ $P$-a.s. for all $t\geq 0$) such that  either 
\begin{eqnarray}\label{NickYor}
\lim_{t\longrightarrow+\infty} N_t=0,\ P\mbox{-a.s.}\quad\mbox{and}\quad G_t=\displaystyle{{N_t}\over{\displaystyle\sup_{0\leq s\leq t}N_s}},\quad\forall\ t\geq 0,\end{eqnarray}
or 
\begin{eqnarray}\label{Model1}
\lim_{t\longrightarrow+\infty} N_t=0,\ P\mbox{-a.s.}\quad\mbox{and}\quad G_t=\min(1, N_t),\quad \forall\quad t\geq 0.\end{eqnarray}
Then $\tau<+\infty$ $P$-a.s. and the condition (\ref{NSC4NFLVRbis}) always fails for these models.
\end{lemma}
The proof of this lemma is relegated to Appendix \ref{SectionC}. The model of $\tau$ given by (\ref{NickYor}), without the positivity assumption on $N$, goes back to \cite{Nik2005} where it is investigated under the assumptions (\ref{Fontana2}), (\ref{Fontana4}) and $\mathbb F$ is the Brownian filtration. The results of \cite{Fontana} were stated for this model in (\ref{NickYor}), and are based essentially on \cite{Nik2005}. Examples of random times that belong to the class of (\ref{Model1}) can be found in \cite[Example 4.5]{Fontana} and in \cite[Proposition 4.5]{ACDJ0}. This shows that our Theorem \ref{NFLVR} is consistent with the literature, and more importantly it addresses the topic in the most general setting possible. Thus, in virtue of Lemma \ref{LemmaNickYor}, Theorem \ref{NFLVR}-(c) introduces a large class of models of $\tau$ for which the NFLVR concept is stable after stopping with $\tau$ on the one hand. On the other hand, Theorem \ref{NFLVR}-(c) shows that ``classical" arbitrages might be be triggered by another factor other than the positivity of $G$ that is intimately related to the failure of NUPBR (see \cite{ACDJ1, ACDJ3}). Recall that NFLVR is violated if and only if there are classical arbitrages or NUPBR is violated, see \cite{DelbaenSchachermayer94, DelbaenSchachermayer99} for details. Thus, in our view, Section \ref{SectionOfDeflators} and Theorem \ref{NFLVR} lunch a method that deals with NFLVR fully. In fact, we believe that the following {\it conjecture} is true.
\begin{eqnarray*}
(\tau, \mathbb F) \mbox{ satisfies (\ref{NSC4NFLVR}) iff NFLVR is stable under stopping with $\tau$.}\end{eqnarray*}
 This projected full answer constitutes our future work, as it requires novel and heavy developments for $(\tau, \mathbb F)$ that are beyond our current scope.
 \subsection{A variant of NFLVR}
    In contrast to NFLVR which is not altered by an equivalent change of probability, this subsection addresses a variant of NFLVR that is very sensitive to any change of probability.
    
\begin{definition} Let $q\in (1,+\infty)$. A model $(X, \mathbb H, Q)$ satisfies $q$-NFLVR, or $X$ satisfies $q$-NFLVR($\mathbb H, Q)$, if there exists a local martingale deflator for $(X, \mathbb H, Q)$ that is a $q$-integrable $\mathbb H$-martingale, i.e., $ {\cal Z}^{(q)} (X, \mathbb H, Q)\not=\emptyset$, where
    \begin{eqnarray}\label{Zp}
    {\cal Z}^{(q)} (X, \mathbb H, Q):={\cal Z}_{loc}(X, \mathbb H, Q)\cap {\cal M}^q(\mathbb H, Q).\end{eqnarray} When $Q=P$, we simply omit the probability for simplicity.\end{definition}
For this $q$-NFLVR concept and its applications, we refer the reader to \cite{ChoulliStricker99,ChoulliStricker98,DelbaenSchachermayer96a,DelbaenSchachermayer96b, Stricker90} and the references therein to cite a few. 

    \begin{theorem}\label{pNFLVR} Suppose $G>0$. Let $T\in (0,+\infty)$ be a  fixed horizon, $q\in (1,+\infty)$, and put $\tilde{q}:=2q/(1+q)$. Then the following assertions hold.\\
  {\rm{(a)}} $S^{\tau\wedge T}$ satisfies $q$-NFLVR $(\mathbb G, \widetilde Q_T)$ as soon as $S^ T$ satisfies $q$-NFLVR $(\mathbb F, P)$, where $\widetilde Q_T$ is given by (\ref{Qtilde}). Furthermore, we have
  \begin{eqnarray}\label{Z(S, G)}
  \left\{Z^{\tau\wedge T}:\quad  Z\in {\cal Z}^{(q)}(S^T, \mathbb F)\right\}\subset {\cal Z}^{(q)}(S^{\tau\wedge T}, \mathbb G, \widetilde Q_T).\end{eqnarray}
    {\rm{(b)}} If $S^ T$ satisfies $q$-NFLVR $(\mathbb F, P)$, then $S^ {\tau\wedge T}$ satisfies $\tilde{q}$-NFLVR $(\mathbb G, P)$. More precisely, we have 
    \begin{eqnarray}\label{Zp(F,G)}
    \left\{ {{Z^{\tau\wedge T}}\over{{\cal E}(G_{-}^{-1}\is m)^{\tau\wedge T}}}:\quad Z\in {\cal Z}^{(q)}(S^T, \mathbb F)\right\}\subset {\cal Z}^{(\tilde{q})}(S^{\tau\wedge T}, \mathbb G).\end{eqnarray}
  {\rm{(c)}} Suppose there exists $\theta$, a $D^{o,\mathbb F}$-integrable and $\mathbb F$-adapted process, such that
    \begin{eqnarray}\label{SufficientCondition}
   G\theta\geq 1,\quad \mbox{and}\quad \theta{\cal E}\left(-\theta{{G}\over{\widetilde G}}\is D^{o,\mathbb F}\right)\ \mbox{is a bounded processes}.
    \end{eqnarray}
  If ${\cal Z}^{(q)}(S^T, \mathbb F)\not=\emptyset$, then ${\cal Z}^{(q)}(S^{\tau\wedge T}, \mathbb G)\not=\emptyset$.
   \end{theorem}
  
 On the one hand, it is clear that (\ref{SufficientCondition}) is fulfilled when $\tau$ is a pseudo-stopping time (i.e., $M^{\tau}\in {\cal M}(\mathbb G)$ for any $M\in {\cal M}(\mathbb F)$). Indeed, thanks to   (\ref{MultiDecompo4G}) and \cite{Nik2005}, where it is stated that $\tau$ is a pseudo-stopping time if and only if $m\equiv m_0$, we deduce that $G=G_0{\cal E}(-{\widetilde G}^{-1}\is D^{o,\mathbb F})$. Hence, by considering $\theta=G^{-1}$, we conclude that (\ref{SufficientCondition}) holds. On the other hand, the condition (\ref{SufficientCondition}) follows clearly from either 
 \begin{eqnarray}\label{Condition4G/m}
 G\geq \delta\quad \mbox{or}\quad {\cal E}(G_{-}^{-1}\is m)\geq \delta',\end{eqnarray}
 where $\delta>$ and $\delta'>0$. In fact, by taking $\theta=G^{-1}$, we deduce that $$\theta{\cal E}(-\theta G{\widetilde G}^{-1}\is D^{o,\mathbb F})=G^{-1}{\cal E}(-{\widetilde G}^{-1}\is D^{o,\mathbb F})=G_0^{-1} {\cal E}({G}_{-}^{-1}\is m)^{-1},$$ is bounded by ${\delta}^{-1}$  and by $(G_0\delta')^{-1}$ for the cases $ G\geq \delta$ and $ {\cal E}(G_{-}^{-1}\is m)\geq \delta'$ respectively. Furthermore, $ G\geq \delta$ implies that $ {\cal E}(G_{-}^{-1}\is m)\geq \delta/G_0=:\delta'$, and this proves that the weakest condition in  (\ref{Condition4G/m}) is $ {\cal E}(G_{-}^{-1}\is m)\geq \delta'$. The proof of Theorem \ref{pNFLVR}-(c) is essentially based on the following lemma.
 \begin{lemma}\label{lemma4Zp} Let $q\in (1,+\infty)$, $Z={\cal E}(M)$ be a positive $\mathbb F$-martingale that is $q$-integrable, and ${\cal X}$, ${\cal X}_+^o$ and $\Theta$ be the sets given by 
 \begin{eqnarray}\label{SetXi}
 {\cal X}&&:=\left\{(\varphi^{(o)}, \varphi^{(pr)})\in \Phi_q \ \big|\ \mbox{(\ref{Condition4Fi(pr)}) and (\ref{ineqMultiGeneral1}) hold}\right\},\hskip 1cm\\
{\cal X}_+^o&&:=\left\{\varphi^{(o)}\ \Big|\   (\varphi^{(o)}, 0)\in {\cal X},\ \varphi^{(o)}\geq 0\ P\otimes dD^{o,\mathbb F}\mbox{-a.e.}\right\},\\
\Theta&&:=\left\{\theta\ {\mathbb F}\rm{\mbox{-optional}}\ \big|\quad \theta G\geq 1\quad \mbox{and}\quad {\widetilde G}(\theta- G^{-1})\in {\cal X}_+^o\right\}.
 \end{eqnarray}
 Here $(\varphi, \psi)\in \Phi_q$ means $\varphi\in  {\cal I}_{loc}^o(N^{\mathbb G}, \mathbb G)$ such that $\varphi\is N^{\mathbb G}\in {\cal M}_{loc}^q(\mathbb G)$ and $\psi\in L^q_{loc}(\Omega\times[0,+\infty), \rm{Prog}(\mathbb F),P\otimes dD)$. Then the following inequalities hold
 \begin{eqnarray*}
&&\hskip -0.75cm \inf_{(\varphi^{(o)}, \varphi^{(pr)})\in {\cal X}}E\left[\left(Z_{\tau\wedge T}/{\cal E}_{\tau\wedge T}(G_{-}^{-1}\is m)\right)^q{\cal E}_T(\varphi^{(o)}\is N^{\mathbb G})^q{\cal E}_T(\varphi^{(pr)}\is D)^q\right]\\
 &&\hskip -0.75cm\leq 
 \inf_{\varphi^{(o)}\in {\cal X}_+^o}E\left[\left(Z_{\tau\wedge T}/{\cal E}_{\tau\wedge T}(G_{-}^{-1}\is m)\right)^q{\cal E}_T(\varphi^{(o)}\is N^{\mathbb G})^q\right]\\
 &&\hskip -0.75cm\leq \inf_{\theta\in\Theta}E\left[ \int_0^T G_0^qZ_t^q\theta_t^q{\cal E}_t(-\theta{{G}\over{\widetilde G}}\is D^{o,\mathbb F})^q dD^{o,\mathbb F}_t+G_T^{1-q}Z^q_T{\cal E}_T\left(-\theta{{G}\over{\widetilde G}}\is D^{o,\mathbb F}\right)^q \right].\end{eqnarray*}
  \end{lemma}
 The proof of the lemma is relegated to Appendix \ref{SectionC} for the sake of simple exposition, and the rest of this section focuses on proving Theorem \ref{pNFLVR}.
 \begin{proof}{\it of Theorem \ref{pNFLVR}} The proof of the theorem will be achieved in three parts, where assertions (a), (b) and (c) are proved respectively. \\
  {\bf Part 1.} Let $Z\in {\cal Z}_{loc}(S^T, \mathbb F)$. Then there exists an $\mathbb F$-predictable process $\varphi$ such that $0<\varphi\leq 1$ and $Z(\varphi\is S)^T\in {\cal M}_{loc}(\mathbb F)$. Thus, a direct application of Proposition \ref{caseofGlocalmartingale} to the pair $Z^T$ and  $Z^T(\varphi\is S)^T$ implies that $Z^{\tau\wedge T}/{\cal E}(G_{-}^{-1}\is m)^{\tau\wedge T}$ and $Z^{\tau\wedge T}(\varphi\is S^{\tau\wedge T})/{\cal E}(G_{-}^{-1}\is m)^{\tau\wedge T}$ belong to ${\cal M}_{loc}(\mathbb G)$, or equivalently $Z^{\tau\wedge T}$ and $Z^{\tau\wedge T}(\varphi\is S^{\tau\wedge T})$ belong to ${\cal M}_{loc}(\mathbb G, \widetilde Q_T)$. This proves that $Z^{\tau\wedge T}$ belongs to $ {\cal Z}_{loc}(S^{\tau\wedge T}, \mathbb G, \widetilde Q)$, for any $Z\in {\cal Z}_{loc}(S^T, \mathbb F)$ on the one hand. On the other hand, by putting $\widetilde Q:=\widetilde Q_T$, for any  $Z\in {\cal Z}^{(q)}(S^T, \mathbb F)$, we derive 
  \begin{eqnarray}
  &&E^{\widetilde Q}\left[Z_{\tau\wedge T}^q\right]=E\left[ {{Z_{\tau\wedge T}^q}\over{{\cal E}_{\tau\wedge T}(G_{-}^{-1}\is m)}} \right]\nonumber\\
  &&=E\left[ \int_0^T {{Z_t^q}\over{{\cal E}_t(G_{-}^{-1}\is m)}}dD^{o,\mathbb F}_t \right]+E\left[1-G_0+ {{Z_{ T}^qG_T}\over{{\cal E}_{T}(G_{-}^{-1}\is m)}} \right]\nonumber\\
  &&= E\left[1-G_0- G_0\int_0^T Z_t^q  d{\cal E}_{t}(-  {\widetilde G}^{-1}\is  D^{o,\mathbb F})\right]+E\left[ G_0Z_{ T}^q{\cal E}_T(-  {\widetilde G}^{-1}\is  D^{o,\mathbb F})\right]\nonumber\\
  &&=1+ E\left[ G_0\int_0^T {\cal E}_{t-}(- {\widetilde G}^{-1}\is  D^{o,\mathbb F})dZ_t^q  \right]\nonumber\\
  &&=1+q(q-1) E\left[ G_0 \int_0^T  Z_{t-}^q{\cal E}_{t-}(- {\widetilde G}^{-1}\is  D^{o,\mathbb F})d{\cal H}^{(q)}_t(M)  \right]\nonumber\\
  &&\leq 1+q(q-1)E\left[ \int_0^T  Z_{s-}^q d{\cal H}^{(p)}_s(M) \right]=E\left[Z_T^q  \right],\label{LastEqua}
 \end{eqnarray}
 where $ M:=Z_{-}^{-1}\is Z$ and  ${\cal H}^{(q)}(M)$ is defined by (\ref{Hellingerp}). The last two equalities is a direct consequences of (\ref{HellingerEquality2})-(\ref{HellingerEquality3}), while the second equality follows from (\ref{MultiDecompo4G}) and (\ref{Grelationship2m}). Therefore, (\ref{LastEqua}) proves that  $E^{\widetilde Q}\left[Z_{\tau\wedge T}^q\right]\leq E[Z^q_T]<+\infty$, and hence $Z^{\tau\wedge T}\in {\cal Z}^{(q)}\left(S^{\tau\wedge T}, \mathbb G, \widetilde Q_T\right)$. This ends the proof of assertion (a)  and that of (\ref{Z(S, G)}). \\
  {\bf Part 2.} Here, we prove assertion (b).  To this end, we consider a process $Z$ such that $Z^T\in  {\cal Z}^{(q)}(S^T, \mathbb F)$ and
   put $q_1:=2q/(1+q)\in (1,2)$, $q_2:=(q+1)/(q-1)=1/(q_1-1)$ and $q_3:=q_2/(q_2-1)$. Then we calculate
  \begin{eqnarray*}
  &&E\left[ {{Z_{\tau\wedge T}^{q_1}}\over{{\cal E}_{\tau\wedge T}(G_{-}^{-1}\is m)^{q_1}}} \right]=E^{\widetilde Q}\left[ {{Z_{\tau\wedge T}^{q_1}}\over{{\cal E}_{\tau\wedge T}(G_{-}^{-1}\is m)^{q_1-1}}} \right]\\
 && \leq \left(E^{\widetilde Q}[ Z_{\tau\wedge T}^{q_3q_1}]\right)^{1/q_3}\left(E^{\widetilde Q}[{\cal E}_{\tau\wedge T}(G_{-}^{-1}\is m)^{-1}]\right)^{1/q_2}\\
 &&= \left(E^{\widetilde Q}[ Z_{\tau\wedge T}^{q_3q_1}]\right)^{1/q_3}=\left(E^{\widetilde Q}[ Z_{\tau\wedge T}^{q}]\right)^{1/q_3}\leq \left(E[ Z_{T}^q]\right)^{1/q_3},
  \end{eqnarray*}
  where the last inequality follows from (\ref{LastEqua}). This proves (\ref{Zp(F,G)}) and ends the proof of assertion (b).\\
 {\bf Part 3.} Here, we prove assertion (c).  To this end, we start by remarking that, for any $q\in (1,+\infty)$, ${\cal Z}^{(q)}\left(S^{\tau\wedge T}, \mathbb G\right)\not=\emptyset$ if and only if 
 $$\inf_{Z^{\mathbb G}\in {\cal Z}_{loc}\left(S^{\tau\wedge T}, \mathbb G\right)}E\left[\left(Z^{\mathbb G}_{\tau\wedge T}\right)^q\right]<+\infty.$$
To prove that this latter holds, we consider $Z\in {\cal Z}^{(q)}\left(S^T, \mathbb F\right)$, and  derive 
 \begin{eqnarray*}
 &&\inf_{Z^{\mathbb G}\in {\cal Z}_{loc}\left(S^{\tau\wedge T}, \mathbb G\right)}E\left[\left(Z^{\mathbb G}_{\tau\wedge T}\right)^q\right]\\
 &&\leq  \inf_{(\varphi^{(o)}, \varphi^{(pr)})\in {\cal X}}E\left[\left(Z_{\tau\wedge T}/{\cal E}_{\tau\wedge T}(G_{-}^{-1}\is m)\right)^q{\cal E}_T(\varphi^{(o)}\is N^{\mathbb G})^q{\cal E}_T(\varphi^{(pr)}\is D)^q\right]\\
 &&\leq C^q E\left[Z_{T\wedge\tau}^q+Z_T^q\right]\leq 2C^q E\left[\sup_{0\leq t\leq T}Z_t^q\right]\leq 2\left({{qC}\over{q-1}}\right)^qE[Z^q_T]<+\infty.\end{eqnarray*}
 The first inequality is a direct consequence of Theorem \ref{LocalMartingaleDeflator}-(b), while the second inequality is a combination of Lemma \ref{lemma4Zp}, the assumption (\ref{SufficientCondition}) (i.e., $\theta{\cal E}(-\theta G{\widetilde G}^{-1}\is D^{o,\mathbb F})\leq C$),  and the fact that under the assumption (\ref{SufficientCondition}) we always have $G^{-1} {\cal E}(-\theta G{\widetilde G}^{-1}\is D^{o,\mathbb F})\leq \theta{\cal E}(-\theta G{\widetilde G}^{-1}\is D^{o,\mathbb F})\leq C$. This proves assertion (b) and the proof of the theorem is completed. 
  \end{proof}

\begin{appendices}

\section{Some $\mathbb G$-properties versus those in $\mathbb F$}
This section has three lemmas. We start by recalling a lemma, stated somehow in \cite{ACDJ1,ACDJ3}, that relates $\mathbb G$-compensators to $\mathbb F$-compensators. 

\begin{lemma}\label{lemmaV} Let $V$ be an $\mathbb F$-adapted RCLL process with finite variation. Then the following assertions hold.\\
{\rm{(a)}} If $V \in {\cal A}_{loc} ({\mathbb F})$, then we have $$(V^{\tau})^{p, {\mathbb G}}= I_{\Lbrack 0,\tau\Lbrack} G_-^{-1} \centerdot ({\widetilde G} \centerdot V)^{p, {\mathbb F}}.$$
{\rm{(b)}} Suppose that $G>0$. Then $V\in {\cal A}_{loc}(\mathbb F)$ iff $({\widetilde G}^{-1}\is V)^{\tau}\in {\cal A}_{loc}(\mathbb G)$.
 \end{lemma}
 \begin{proof}
 The proof of assertion (a) can be found in \cite[Lemma 3.1]{ACDJ1}. To prove assertion (b), we remark that $V \in {\cal A}_{loc} ({\mathbb F})$ if and only if $\mbox{Var}(V) \in {\cal A}_{loc} ({\mathbb F})$, where Var$(V)$ is the variation of $V$, and Var$({\widetilde G}^{-1}\is V^{\tau})=({\widetilde G}^{-1}\is \mbox{Var}(V))^{\tau}$. Thus, it is enough to prove assertion (b) for the case when $V$ is nondecreasing. On the one hand, the proof of $V\in {\cal A}_{loc}(\mathbb F)\Longrightarrow ({\widetilde G}^{-1}\is V)^{\tau}\in {\cal A}_{loc}(\mathbb G)$  follows immediately from \cite[Lemma 3.2]{ACDJ1} and the assumption $G>0$. On the other hand, due to \cite[Proposition B.2-(b)]{ACDJ1}, $({\widetilde G}^{-1}\is V)^{\tau}\in {\cal A}_{loc}(\mathbb G)$ induces the existence of an increasing sequence of $\mathbb F$-stopping times $(T_n)_n$ such that $\sup_n T_n\geq R_0$ $P$-a.s., where $R_0$ is given by (\ref{Rzero}), and 
 \begin{eqnarray*}
 E[V_{T_n}]=E\left[({\widetilde G}^{-1}\is V)_{\tau\wedge T_n}\right]<+\infty.\end{eqnarray*}
 Due to the assumption $G>0$ which is equivalent to $R_0=+\infty$ $P$-a.s., see Lemma \ref{Gdecomposition}, we conclude that $T_n$ increases to infinity almost surely, and hence $V\in {\cal A}_{loc}(\mathbb F)$. This ends the proof of the lemma.\end{proof}
The following second lemma is very useful in our analysis, and it has four assertions. The first and the last assertions are already established in the literature, while the second and third sound to be new to us. 
\begin{lemma}\label{PortfolioGtoF}  The following assertions hold.\\
{\rm{(a)}} For any $\mathbb G$-predictable process $\varphi^{\mathbb G}$, there exists an $\mathbb F$-predictable process $\varphi^{\mathbb F}$ such that
  $\varphi^{\mathbb G}I_{\Lbrack0,\tau\Lbrack}=\varphi^{\mathbb F}I_{\Lbrack0,\tau\Lbrack}$. Furthermore, if $\varphi^{\mathbb G}$ is bounded, then $\varphi^{\mathbb F}$ can be chosen bounded with the same constants. \\
{\rm{(b)}} Suppose that $G>0$. Then for any  bounded $\theta\in {{\Theta}(S^{\tau},\mathbb G)}$, then there exists a bounded ${\varphi}\in {\Theta}(S,\mathbb F)$ that coincides with $ {\theta} $ on $\Rbrack0,\tau\Lbrack$.\\
{\rm{(c)}}  Suppose $G>0$, and let $V^{\mathbb G}$ be a RCLL $\mathbb G$-predictable and nondecreasing process with finite values and $(V^{\mathbb G})^{\tau}=V^{\mathbb G}$. Then there exists a unique RCLL nondecreasing with finite values and $\mathbb F$-predictable process, $V$, such that $V^{\mathbb G}=V^{\tau}$. If furthermore $\Delta V^{\mathbb G}<1$, then $\Delta V<1$ holds also.\\
{\rm{(d)}} For any bounded $\mathbb G$-stopping time $\sigma^{\mathbb G}$, there exists a bounded $\mathbb F$-stopping time $\sigma^{\mathbb F}$ such that 
$
\sigma^{\mathbb G}\wedge\tau= \sigma^{\mathbb F}\wedge\tau,\ P$-a.s..
\end{lemma}

\begin{proof}  It is clear that assertion (d) can be found in \cite[XX.75 b)]{DMM}, see also \cite{Jeulin1980}. Remark that  the boundedness condition on $\varphi^{\mathbb G}$, in assertion (a), can be simplified to $0\leq \varphi^{\mathbb G}\leq 1$.  Thus, assertion (a) is a particular case of the general case treated in \cite[Lemma B.1]{ACDJ1} (see also \cite[Lemma 4.4 (b), page 63]{Jeulin1980}), hence its proof will be omitted and we refer the reader to this paper. Thus, the rest of this proof focuses on assertions (b) and (c) in two parts.\\
{\bf Part 1.} Here we prove assertion (b). Consider  a bounded ${\theta}\in {{\Theta}(S^{\tau},\mathbb G)}$. Then $\theta$ is a bounded and $\mathbb G$-predictable process satisfying ${\theta}^{tr}\Delta S^{\tau}>-1$.  Thus, in virtue of assertion (a), there exists a bounded and $\mathbb F$-predictable process $\varphi$ such that 
$$
{\theta}I_{\Rbrack0,\tau\Lbrack}=\varphi I_{\Rbrack0,\tau\Lbrack}.$$
Then by inserting this equality in ${\theta}^{tr}\Delta S^{\tau}>-1$, we deduce that 
$$\varphi^{tr}\Delta S I_{\Lbrack0,\tau\Lbrack}>-1,$$
which is equivalent to $I_{\Lbrack0,\tau\Lbrack}\leq I_{\{\varphi^{tr}\Delta S>-1\}}$. By taking the $\mathbb F$-optional projection on both sides of this inequality, we get $0<G\leq I_{\{\varphi^{tr}\Delta S>-1\}}$ on $\Lbrack0,+\infty\Rbrack$, or equivalently $\varphi^{tr}\Delta S>-1$. Hence $\varphi$ belongs to ${\Theta}(S,\mathbb F)$, and the proof of assertion (b) is complete. \\
{\bf Part 2.}  This part proves assertion (c). Consider a $\mathbb G$-predictable and nondecreasing process with finite values $V^{\mathbb G}$  such that $(V^{\mathbb G})^{\tau}=V^{\mathbb G}$. It is clear that there is no loss of generality in assuming that $V^{\mathbb G}$  is bounded. Hence,  a direct application of assertion (a) to $V^{\mathbb G}$ implies the existence of an $\mathbb F$-predictable process $V$ such that 
\begin{eqnarray}\label{VGVF}
V^{\mathbb G}I_{\Rbrack0,\tau\Lbrack}=VI_{\Rbrack0,\tau\Lbrack}.\end{eqnarray}
Remark that $V^{\mathbb G}I_{\Rbrack0,\tau\Rbrack}=V^{\mathbb G}- V^{\mathbb G}_{\tau}I_{\Rbrack\tau,+\infty\Rbrack}$ is a RCLL and bounded $\mathbb G$-semimartingale. Thus, by multiplying both sides of (\ref{VGVF}) with $I_{\Rbrack0,\tau\Rbrack}$ and taking the $\mathbb F$-optional projection afterwards, we get $V=^{o,\mathbb F}(V^{\mathbb G}I_{\Rbrack0,\tau\Rbrack})/G$. Hence, the predictable process $V$ is a RCLL $\mathbb F$-semimartingale, as it is the product of the two RCLL $\mathbb F$-semimartingales $G^{-1}$ and $^{o,\mathbb F}(V^{\mathbb G}I_{\Rbrack0,\tau\Rbrack})$ \footnote{Thanks to \cite[Th\'eor\`eme 47, p. 119 and Th\'eor\`eme 59, p. 268]{DellacherieMeyer80}, the optional projection of a bounded RCLL $\mathbb G$-semimartingale is a RCLL $\mathbb F$-semimartingale}. As a result, $V$ is a special semimartingale, and hence there exist $L\in {\cal M}_{loc} (\mathbb F)$ and an $\mathbb F$-predictable process $B\in {\cal A}_{loc}(\mathbb F)$ such that $V=L+B$ and $L_0=0$. By using this equality and by stopping both terms in (\ref{VGVF}) at $\tau$, we derive   
\begin{eqnarray}
V^{\mathbb G}&&=V^{\tau}=L^{\tau}+B^{\tau}\label{VG=Vtau}\\
&&=\left(L^{\tau}-G_{-}^{-1}I_{\Lbrack0,\tau\Lbrack}\is \langle L,m\rangle^{\mathbb F}\right)+G_{-}^{-1}I_{\Lbrack0,\tau\Lbrack}\is \langle L,m\rangle^{\mathbb F}+B^{\tau}.\nonumber
\end{eqnarray}
As $V^{\mathbb G}$ is predictable with finite variation, from this latter equality, we conclude that the $\mathbb G$-local martingale $L^{\tau}-G_{-}^{-1}I_{\Lbrack0,\tau\Lbrack}\is \langle L,m\rangle^{\mathbb F}$ is null.  Notice that $L$ is a predictable local martingale, and hence it is continuous (due to $\Delta L=\ ^{p,\mathbb F}(\Delta L)=0$, see \cite[Th\'eor\`eme 32, page 99]{DellacherieMeyer80}). By combining these remarks, we conclude that  ${\cal T}(L)=L^{\tau}-G_{-}^{-1}I_{\Lbrack0,\tau\Lbrack}\is \langle L,m\rangle^{\mathbb F}$ and $[L,L]^{\tau}=[{\cal T}(L), {\cal T}(L)]$ is a null process. This is equivalent to $G_{-}\is [L,L]=([L,L]^{\tau})^{p,\mathbb F}\equiv 0$, or equivalently $L\equiv 0$ due to $G>0$ and Lemma \ref{Gdecomposition}-(a). This proves that  $V=B$ has a finite variation. To prove that $V$ is nondecreasing it is enough to consider the $\mathbb F$-dual predictable projection in both sides of (\ref{VG=Vtau}) and get $(V^{\mathbb G})^{p,\mathbb F}=G_{-}\is V$ or equivalently $V=G_{-}^{-1}\is (V^{\mathbb G})^{p,\mathbb F}$.  This proves the first statement of assertion (c), while the proof of the last statement of assertion (c) follows the same footsteps of part 1. Indeed, by using (\ref{VG=Vtau}), we conclude that $\Delta V^{\mathbb G}=\Delta V I_{\Lbrack0,\tau\Lbrack}<1$ holds if and only if  $I_{\Lbrack0,\tau\Lbrack}\leq I_{\{\Delta V<1\}}$ holds. This latter fact implies that, after taking the $\mathbb F$-predictable projection on both sides of this inequality, $0<G_{-}\leq  I_{\{\Delta V<1\}}$ on $\Lbrack0,+\infty\Rbrack$, or equivalently $\Delta V<1$. The condition $G_{-}>0$ follows from the assumption $G>0$ due to Lemma \ref{Gdecomposition}-(a). This ends the proof of the lemma.
\end{proof}
The following is the last lemma of this section. It is useful in simplifying the proof(s) of Section \ref{SectionOfDeflators}, and sounds important in itself. 
\begin{lemma}\label{Lemma4Wprocess} Let  $\varphi$ be a bounded and $\mathbb F$-predictable process, $L\in {\cal M}_{loc}(\mathbb F)$, $(\varphi^{(o)}, \varphi^{(pr)}) \in {\cal I}^o_{loc}(N^{\mathbb G},\mathbb G)\times L^1_{loc}({\rm{Prog}}(\mathbb F),P\otimes dD)$ satisfying (\ref{Condition4Fi(pr)}) and 
\begin{eqnarray}\label{Positivity1}
\left(1+{{\Delta L}\over{G_{-}\widetilde G}}\right)I_{\Lbrack0,\tau\Lbrack}+\varphi^{(o)}\Delta N^{\mathbb G}+\varphi^{(pr)}\Delta D >0.\end{eqnarray}
 Then the process
$$
W:=\sum \varphi\Delta S I_{\{\vert \Delta S\vert>1\}}\left[\left(1+{{\Delta L}\over{G_{-}\widetilde G}}\right)I_{\Lbrack0,\tau\Lbrack}+\varphi^{(o)}\Delta N^{\mathbb G}+\varphi^{(pr)}\Delta D \right]$$
has a $\mathbb G$-locally integrable variation if and only if both processes
\begin{eqnarray*}
W^{(1)}&&:=\sum \varphi\Delta S I_{\{\vert \Delta S\vert>1\}}(1+{{\Delta L}\over{G_{-}\widetilde G}})I_{\Lbrack0,\tau\Lbrack}\quad\mbox{and}\\
W^{(2)}&&:=\sum \varphi\Delta S I_{\{\vert \Delta S\vert>1\}}[\varphi^{(o)}\Delta N^{\mathbb G}+\varphi^{(pr)}\Delta D ]
\end{eqnarray*}
belong to $ {\cal A}_{loc}(\mathbb G)$.
\end{lemma} 

\begin{proof} Due to (\ref{Positivity1}),  it is clear that  $W\in {\cal A}_{loc}(\mathbb G)$ iff 
$$W^+:=\sum \vert\varphi\Delta S\vert I_{\{\vert \Delta S\vert>1\}}[(1+{{\Delta L}\over{G_{-}\widetilde G}})I_{\Lbrack0,\tau\Lbrack}+\varphi^{(o)}\Delta N^{\mathbb G}+\varphi^{(pr)}\Delta D ]\in {\cal A}_{loc}(\mathbb G).$$
By stopping, there is no loss of generality to assume that $E[W^+_{\infty}]<+\infty$. Thus, due to the boundedness of the $\mathbb F$-optinal process $\varphi\Delta S I_{\{k\geq\vert \Delta S\vert>1\}}$ and $(\varphi^{(o)}, \varphi^{(pr)}) \in {\cal I}^o_{loc}(N^{\mathbb G},\mathbb G)\times L^1_{loc}({\rm{Prog}}(\mathbb F),P\otimes dD)$, we conclude that $$\vert\varphi\Delta S\vert I_{\{k\geq\vert \Delta S\vert>1\}}\left(\varphi^{(o)},\varphi^{(pr)}\right)\in{\cal I}^o_{loc}(N^{\mathbb G},\mathbb G)\times L^1_{loc}({\rm{Prog}}(\mathbb F),P\otimes dD).$$
Hence, both $\sum \vert\varphi\Delta S\vert I_{\{k\geq\vert \Delta S\vert>1\}}\varphi^{(o)}\Delta N^{\mathbb G} =\vert\varphi\Delta S\vert I_{\{k\geq\vert \Delta S\vert>1\}}\varphi^{(o)}\is N^{\mathbb G} $ and $\sum \vert\varphi\Delta S\vert I_{\{k\geq\vert \Delta S\vert>1\}}\varphi^{(pr)}\Delta D=\vert\varphi\Delta S\vert I_{\{k\geq\vert \Delta S\vert>1\}}\varphi^{(pr)}\is D$ are $\mathbb G$-local martingales. Then we consider a sequence of $\mathbb G$-stopping times $(T_n)_n$ that increases to infinity such that both $(\varphi^{(o)}\is N^{\mathbb G})^{T_n}$ and $(\varphi^{(pr)}\is D)^{T_n}$ are uniformly integrable $\mathbb G$-martingales, and we derive 
\begin{eqnarray*}
&&E\left[\sum \vert\varphi\Delta S\vert I_{\{\vert \Delta S\vert>1\}}[(1+{{\Delta L}\over{G_{-}\widetilde G}})I_{\Lbrack0,\tau\Lbrack}\right]\\
&&= \lim_{k, n\longrightarrow+\infty}E\left[\sum \vert\varphi\Delta S\vert I_{\{1<\vert \Delta S\vert\leq k\}}(1+{{\Delta L}\over{G_{-}\widetilde G}})I_{\Lbrack0,\tau\wedge T_n\Lbrack}\right] \\
&&=\lim_{k, n\longrightarrow+\infty}E[(I_{\{\vert\Delta S\vert\leq k\}}\is W^+)_{T_n}]\leq E[W^+_{\infty}]<+\infty.\end{eqnarray*}
This proves that  $W^{(1)}\in {\cal A}_{loc}(\mathbb G)$, and hence $W^{(2)}=W-W^{(1)}\in {\cal A}_{loc}(\mathbb G)$. Thus, the proof of the lemma is complete.
\end{proof}
\section{ Proofs of Lemma \ref{Gdecomposition} and Proposition \ref{GeneralSupDeflators}}\label{AppB}
\begin{proof}{\it of  Lemma \ref{Gdecomposition}:}
It is clear that, in virtue of $\inf(\emptyset)=+\infty$ by convention, the equivalences between $G>0$, $G_{-}>0$ and $\widetilde G>0$ follow immediately from the equalities in (\ref{Rzero}) (see Lemma \ref{Gdecomposition}). Furthermore, these inequalities in (\ref{Rzero}) go back to \cite[Theorem 14, Chapter XX, p. 134]{DMM}. This proves assertion (a). To prove assertion (b), we suppose $G>0$, and  we conclude that both $G_{-}$ and $\widetilde G$ are positive (due to assertion (a)) and derive
  \begin{eqnarray}\label{SDE4G}
  {{dG}\over{G_{-}}}=  {{dm}\over{G_{-}}}-  {{dD^{o,\mathbb F}}\over{G_{-}}},\quad G_0:=P(\tau>0\big| {\cal F}_0).
  \end{eqnarray}
  Then, in virtue of (\ref{DDequation}), we obtain
  \begin{eqnarray*}
  G=G_0{\cal E}\left( {1\over{G_{-}}}\is m-  {1\over{G_{-}}}\is D^{o,\mathbb F}\right).\end{eqnarray*}
  Thanks to $\Delta m={\widetilde G}-G_{-}$ and (\ref{YorFormula}) (which implies ${\cal E}(X+V)={\cal E}(X){\cal E}((1+\Delta X)^{-1}\is V)$ for any semimartingale $X$ satisfying $1+\Delta X>0$ and any RCLL process with finite variation $V$), the proof of assertion (b) follows.\\
 Thanks to \cite[Th\'eor\`eme 6, Chapitre VI, p.79]{DellacherieMeyer80}, that claims that any nonnegative RCLL supermatingale admits the limit at infinity $P$-almost surely, we deduce that the first claim of assertion (c) always holds. The equality (\ref{Ginfinity}) follows from this claim and (\ref{MultiDecompo4G}). As $G_{\infty}$ is ${\cal F}_{\infty-}$-measurable due to (\ref{Ginfinity}), where ${\cal F}_{\infty-}:=\sigma(\cup_{t\geq 0}{\cal F}_t)$, we derive 
$$P\left((\tau=\infty)\cap(G_{\infty}=0)\right)=E\left[G_{\infty}I_{\{G_{\infty}=0\}}\right]=0.$$ This proves that $(\tau=+\infty)\subset(G_{\infty}>0)$, while the equality in (\ref{inclusion4G}) is straightforward due to $G_0>0$ and (\ref{Ginfinity}). This proves assertion (d). Assertion (e) is a direct consequence of $E[G_{\infty}]=P(\tau=+\infty)$ and the equality in (\ref{inclusion4G}). This ends the proof of the lemma. \end{proof}   
 The proof of  Proposition \ref{GeneralSupDeflators} is based on the following remark
 \begin{remark}\label{MDS}
 If $Z$ is a deflator for $(X,\mathbb H)$, then there exists a unique pair $(N, V)\in {\cal M}_{loc}(\mathbb H)\times {\cal A}_{loc}(\mathbb H)$ such that $V$ is nondecreasing $\mathbb H$-predictable, and  (\ref{MultiDecompoDeflator})  holds. \end{remark}
 In fact when $Z$ is a deflator for $(X,\mathbb H)$, it is clear that $Z$ is a positive supermartingale, as $\varphi=0\in {\Theta}(X,\mathbb H)$. Hence, a direct application  of \cite[Theorem 8.21 page 138]{JS03}, see also \cite{Azema} and \cite[Proposition 1.32, page 15]{AJ}, we conclude the existence of the pair $(N, V)$ described in the remark.\\

\begin{proof}{\it of  Proposition \ref{GeneralSupDeflators}:}  This proof will be divided into two steps. The first step shows that for a process $Z$, for which there exists a pair $(N,V)$ satisfying (\ref{MultiDecompoDeflator}), there is equivalence between $Z{\cal E}(\varphi\is X)$ being supermartingale and (\ref{deflatorAssumptions}) for any bounded $\varphi\in{\Theta}(X,\mathbb H)$. Thus, it is clear that the combination of this with Remark \ref{MDS} proves assertion (a), while assertion (b) will be proved in the second step. \\
{\bf Step 1.} Suppose that there exists a pair $(N,V)$ such that $Z=Z_0{\cal E}(N){\cal E}(-V)$ and  (\ref{MultiDecompoDeflator})  holds.  Let $\varphi$ be a bounded element of ${\Theta}(X,\mathbb H)$. Then by applying Yor's formula repeatedly, see (\ref{YorFormula}),  one get 
   \begin{eqnarray*}
Z{\cal E}(\varphi \is X) &&= Z_0 {\cal E}(N) {\cal E}(\varphi \is X) {\cal E}( - V) =  Z_0 {\cal E}\Bigl(N +\varphi \is X + \varphi \is [ X, N]\Bigr) {\cal E}( - V) \\
&&=  Z_0 {\cal E}(Y^{(\varphi)} ) {\cal E}(-V)= Z_0 {\cal E}\left( Y^{(\varphi)}  - V - [ Y^{(\varphi)} , V]\right) \\
  &&= Z_0 {\cal E}\left( ( 1 - \Delta V) \is Y^{(\varphi)}  - V\right),
 \end{eqnarray*}
  where $Y^{(\varphi)}:= N +\varphi \is X + \varphi \is [ X, N].$ As $Z$  is positive and $\varphi\in{\Theta}(X,\mathbb H)$ and in virtue of the above equality, we deduce that $Z{\cal E}(\varphi \is X) $ is an ${\mathbb H}$-supermartingale if and only if  $( 1 - \Delta V) \is Y^{(\varphi)} - V$ is a local ${\mathbb H}$-supermartingale, or equivalently $Y^{(\varphi)}$ is a special semimartingale (which is equivalent to the first condition of (\ref{deflatorAssumptions})) and its predictable with finite variation part, $A^{(\varphi,N,\mathbb H)}$, satisfies $A^{(\varphi,N,\mathbb H)}\preceq (1-\Delta V)^{-1}\is V$. This finishes the first step. \\
  {\bf Step 2.}  On the one hand, thanks to Remark \ref{MDS}, when $Z$ is a local martingale deflator, then there exists a unique $N\in {\cal M}_{loc}(\mathbb H)$, with $N_0=0$, for which the first and the second conditions in (\ref{MartingaleDeflator1}) hold. On the other hand, similar calculations to those of step 1 applied to $Z(\varphi\is X)$ by putting $V\equiv 0$ and substituting $\preceq$ by $=$, for $\varphi$ predictable and $0<\varphi\leq 1$, we conclude that $Z_0{\cal E}(N)(\varphi\is X)$ is a local martingale if and only if both the last condition in (\ref{MartingaleDeflator1}) and (\ref{MartingaleDeflator2}) hold. This proves assertion (b), and completes the proof of the proposition.
\end{proof}
\section{Proof of Lemmas \ref{keylemma}, \ref{LemmaNickYor} and \ref{lemma4Zp}}\label{SectionC}
We start by proving Lemma \ref{keylemma}. 
\begin{proof}{\it of Lemma \ref{keylemma}}: As $N$ is a uniformly integrable $\mathbb F$-martingale, then there exists $N_{\infty}\in L^1({\cal F}_{\infty-}, P)$ such that $ N={^{o,\mathbb F}}(N_{\infty})$. Hence, we derive 

\begin{eqnarray*}
E\left[\int_0^{\infty} -\vert N_s\vert d{\cal E}_s(-{1\over{\widetilde G}}\is D^{o,\mathbb F})\right]&&\leq E\left[\int_0^{\infty} -E\left[\vert N_{\infty}\vert \big|{\cal F}_s\right]d{\cal E}_s(-{1\over{\widetilde G}}\is D^{o,\mathbb F})\right]\\
&&= E\left[\vert N_{\infty}\vert (1-{\cal E}_{\infty}(-{\widetilde G}^{-1}\is D^{o,\mathbb F}))\right]<\infty. \end{eqnarray*}
This proves the first statement of the lemma. To prove the second statement, we use the integration by part formula and the facts that  ${\cal E}(-{\widetilde G}^{-1}\is  D^{o,\mathbb F})\leq 1$ and $\vert N\vert \leq {^{o,\mathbb F}}(\vert N_{\infty}\vert)$, and write 
\begin{eqnarray*}
 \left\vert \int_0^t {\cal E}_{s-}(-{1\over{\widetilde G}}\is  D^{o,\mathbb F})dN_s\right\vert\leq E\left[\vert N_{\infty}\vert +\vert N_0\vert-\int_0^{\infty}\vert N_s\vert d{\cal E}_{s}(- {1\over{\widetilde G}}\is  D^{o,\mathbb F})\Big|{\cal F}_t\right].
\end{eqnarray*}
By combining this with the first statement, the proof of the lemma follows.
\end{proof}
\begin{proof}{\it of Lemma \ref{LemmaNickYor}:}
1) Suppose that the pair $(\tau, \mathbb F)$ satisfies (\ref{NickYor}). Without loss of generality, one can assume $N_0=1$ in this case, and derive that $0<G\leq N$. Thus, due to the first condition in (\ref{NickYor}), we get $G_{\infty}=0$ $P$-a.s., which is equivalent to $\tau<+\infty$ $P$-a.s. in virtue of Lemma \ref{Gdecomposition}. On the one hand, thanks to the continuity of $N$ and the uniqueness of the canonical decomposition of $G$, we get 
 \begin{eqnarray*}
 {\cal E}_t(G_{-}^{-1}\is m)=N_t\quad\mbox{and}\quad {\cal E}_t(-{\widetilde G}^{-1}\is D^{o,\mathbb F})=\left(\sup_{0\leq s\leq t} N_s\right)^{-1}.\end{eqnarray*}
  On the other hand, due to the first condition in (\ref{NickYor}) again, it is clear that $\displaystyle\sup_{t\geq 0} N_t<+\infty$ $P$-a.s., and hence (\ref{NSC4NFLVRbis}) fails.\\
  2) Suppose that (\ref{Model1}) holds. Hence, again it is clear that in this case we have $0<G\leq N$. Thus, due to the first condition in (\ref{Model1}), we get $G_{\infty}=0$ $P$-a.s., which is equivalent to $\tau<+\infty$ $P$-a.s., see Lemma \ref{Gdecomposition}-(e) for details.  Thanks to Tanaka's formula, see \cite[Theorem 1.2, Chapter VI]{Yorbook}, and the fact that $\min(a, b)=b-(a-b)^-$, we derive 
  \begin{eqnarray*}
  G_t=G_0+\int_0^t I_{\{N_s\leq 1\}}d N_s-{1\over{2}} L^1_t, \end{eqnarray*}
  where $L^1$ is the local time of $N$ in $1$. Then the continuity of $G$ and $N$, and the uniqueness of the canonical decomposition of $G$ lead to
  \begin{eqnarray*}
  m=G_0+I_{\{N\leq 1\}}\is N\quad\mbox{and}\quad D^{o,\mathbb F}={1\over{2}} L^1. 
  \end{eqnarray*}
 It is also clear that, due to the continuity of $G$, we have $G={\widetilde G}=G_{-}$. Furthermore, thanks to \cite[Proposition 1.3, Chapter VI]{Yorbook}, the nondecreasing process $L^1$ is almost surely carried by $\{t\geq 0:\ N_t=1\}$, and hence we get 
  \begin{eqnarray*}
  {\cal E}(-{\widetilde G}^{-1}\is D^{o,\mathbb F})=\exp\left(-{1\over{2}} L^1\right).
  \end{eqnarray*}
  A combination of \cite[Theorem 1.2, Chapter VI]{Yorbook} and Fatou, leads to 
  \begin{eqnarray*}
  E[L^1_{\infty}]\leq \lim_{t\longrightarrow+\infty} E[L^1_t]\leq E[N_0+1]<+\infty.\end{eqnarray*}
   This implies that $L^1_{\infty}<+\infty$ $P$-a.s., or equivalently ${\cal E}_{\infty}(-{\widetilde G}^{-1}\is D^{o,\mathbb F})>0$ $P$-a.s.. This proves that (\ref{NSC4NFLVRbis}) fails for the model in (\ref{Model1}), and the proof of the lemma is completed.
 \end{proof}
The rest of this section proves Lemma  \ref{lemma4Zp}.

\begin{proof}{\it of Lemma  \ref{lemma4Zp}:} Consider $Z={\cal E}(M)$ a $q$-integrable $\mathbb F$-martingale, and $T\in (0,+\infty)$, and hence $Z^T\in {\cal M}^{(q)}(\mathbb F)$. As ${\cal X}^o_+\times \{0\}\subset {\cal X}$, the first inequality of the lemma is obvious. Let $\varphi^{(o)}\in {\cal X}_+^o$, and remark that $\theta:=G^{-1}+\varphi^{(o)}/{\widetilde G}\in\Theta$ if and only if  $\varphi^{(o)}\in {\cal X}_+^o$. Furthermore, by combining (\ref{YorFormula}) , $[D^{o,\mathbb F}, D^{o,\mathbb F}]=\Delta D^{o,\mathbb F}\is D^{o,\mathbb F}$ and $\Delta D^{o,\mathbb F}=\widetilde G -G$ we derive
\begin{eqnarray}\label{varphi00}
&&{\cal E}(-{{\varphi^{(o)}}\over{\widetilde G}}I_{\Lbrack0,\tau\Rbrack}\is D^{o,\mathbb F}){\cal E}(-{1\over{\widetilde G}}I_{\Lbrack0,\tau\Rbrack}\is D^{o,\mathbb F})\nonumber\\
&&={\cal E}\left( -{{\varphi^{(o)}}\over{\widetilde G}}I_{\Lbrack0,\tau\Rbrack}\is D^{o,\mathbb F}-{1\over{\widetilde G}}I_{\Lbrack0,\tau\Rbrack}\is D^{o,\mathbb F}+{{\varphi^{(o)}}\over{{\widetilde G}^2}}I_{\Lbrack0,\tau\Rbrack}\is [D^{o,\mathbb F},D^{o,\mathbb F}]\right)\nonumber\\
&&={\cal E}\left( {{-{\widetilde G}\varphi^{(o)}-{\widetilde G}+\varphi^{(o)}\Delta D^{o,\mathbb F}}\over{{\widetilde G}^2}}I_{\Lbrack0,\tau\Rbrack}\is D^{o,\mathbb F}\right)\nonumber\\
&&={\cal E}(-{{\theta G}\over{\widetilde G}}I_{\Lbrack0,\tau\Rbrack}\is D^{o,\mathbb F}),
 \end{eqnarray}
 and due to direct calculations we obtain 
 \begin{eqnarray}\label{varphi01}
{\cal E}(\varphi^{(o)}\is N^{\mathbb G})^q&&\hskip -0.5cm={\cal E}(-\varphi^{(o)}{\widetilde G}^{-1}I_{\Lbrack0,\tau\Rbrack}\is D^{o,\mathbb F})^q(1+\varphi^{(o)}{{G}\over{\widetilde G}}D)^q\nonumber\\
&&\hskip -0.5cm={\cal E}\left(-{{\theta G}\over{\widetilde G}}I_{\Lbrack0,\tau\Rbrack}\is D^{o,\mathbb F}\right)^q{{(\theta G)^q D+1-D}\over{{\cal E}(-{\widetilde G}^{-1}I_{\Lbrack0,\tau\Rbrack}\is D^{o,\mathbb F})^q}}.\hskip 0.5cm
 \end{eqnarray}
By combining (\ref{varphi00}), (\ref{varphi01}) and ${\cal E}(G_{-}^{-1}\is m)^{-1}={{G_0}\over{G}}{\cal E}(-{\widetilde G}^{-1}\is D^{o,\mathbb F})$ given by (\ref{Grelationship2m}), we get  
 \begin{eqnarray*}
&&\left({{Z_{\tau\wedge T}}\over{{\cal E}_{\tau\wedge T}(G_{-}^{-1}\is m)}}\right)^q{\cal E}_T(\varphi^{(o)}\is N^{\mathbb G})^q\\
&&\hskip -0.65cm=G_0^q Z_{\tau}^q{\cal E}_{\tau-}(- {{\theta G}\over{\widetilde G}}\is D^{o,\mathbb F})^q\left({{\theta_{\tau} G_{\tau} }\over{\widetilde G_{\tau} }} \right)^qI_{\{\tau\leq T\}}+G_0^q{{Z_T^q}\over{G_T^q}}{\cal E}_{T}(-{{\theta G}\over{\widetilde G}}\is D^{o,\mathbb F})^q I_{\{\tau> T\}}.
 \end{eqnarray*} 
Therefore, by taking expectation on both sides of this equation, we obtain
 \begin{eqnarray*}
 &&E\left[\left(Z_{\tau\wedge T}/{\cal E}_{\tau\wedge T}(G_{-}^{-1}\is m)\right)^q{\cal E}_T(\varphi^{(o)}\is N^{\mathbb G})^q\right]\\
 &&=1+E\left[\int_0^T G_0^q Z_s^q{\cal E}_{s-}\left(-{{\theta G}\over{\widetilde G}}\is D^{o,\mathbb F}\right)^q\left({{\theta_s G_s}\over{\widetilde G_s}}\right)^q dD^{o,\mathbb F}_s\right]\\
 &&+ E\left[G_0^q Z_T^q(G_T)^{1-q}{\cal E}_{T}(-{{\theta G}\over{\widetilde G}}\is D^{o,\mathbb F})^q\right]
  \end{eqnarray*}
  This ends the proof of the lemma. \end{proof}
  \section{A lemma on Hellinger processes}This section extends slightly some results of \cite{Choulli2007,Choulli2009}, on Hellinger processes. These results are useful for the proof of Theorem \ref{pNFLVR}.
\begin{lemma}\label{Choulli2007}
For any $q\in (1,+\infty)$ and any $M\in {\cal M}_{loc}(\mathbb H, P)$ such that $1+\Delta M\geq 0$, we consider the following process
\begin{eqnarray}\label{Hellingerp}
{\cal H}^{(q)}(M):={{1}\over{2}}\langle M^c\rangle +\sum\left[(1+\Delta M)^q-1-q\Delta M\right]/q(q-1).\end{eqnarray}
Then the following assertions hold.\\
{\rm{(a)}} ${\cal H}^{(q)}(M)$ is a RCLL and nondecreasing process with finite values, and ${\cal H}^{(q)}(M)\in {\cal A}_{loc}^+(\mathbb H, P)$ if and only if $M\in {\cal M}_{loc}^{(q)}(\mathbb H, P)$.\\
{\rm{(b)}} The following equalities always hold. 
\begin{eqnarray}
&&{\cal E}(M)^q=1+q{\cal E}_{-}(M)^q\is M +q(q-1){\cal E}_{-}(M)^q\is {\cal H}^{(q)}(M), \label{HellingerEquality1}\\
&&E\left[{\cal E}_T(M)^q\right]=1+q(q-1)E\left[\int_0^T {\cal E}_{s-}(M)^qd {\cal H}_s^{(q)}(M)\right],\label{HellingerEquality2}\\
&&E\left[(h\is {\cal E}(M)^q)_T\right]=q(q-1)E\left[\int_0^T h_s{\cal E}_{s-}(M)^qd {\cal H}_s^{(q)}(M)\right],\label{HellingerEquality3}\end{eqnarray}
for any nonnegative and bounded $\mathbb H$-predictable process $h$. \end{lemma} 
\begin{proof} The proof of  assertion (a) and (\ref{HellingerEquality1}) can be found in \cite[Proposition 3.3 and Proposition 3.8]{Choulli2007} and their proofs, see also \cite[Proposition 1]{Choulli2009} and its proof. The equality (\ref{HellingerEquality2}), when $M\in {\cal M}^{(q)}(\mathbb H, P)$ is also borrowed from these references, while (\ref{HellingerEquality2}) is a direct consequence of (\ref{HellingerEquality3}) by taking $h=1$ for the general case of $M\in {\cal M}_{loc}(\mathbb H, P)$ with $1+\Delta M\geq 0$. Thus, herein, we prove that (\ref{HellingerEquality3}) still holds for this general  case. To this end, we consider a nonnegative and bounded $\mathbb H$-predictable process $h$ and a sequence of stopping times  $(T_n)_{n\geq 1}$ that increases to infinity such that, for all $n\geq 1$, $\left({\cal E}_{-}(M)^q\is M\right)^{T_n}$ is a martingale satisfying $E\left[\sup_{0\leq s\leq T_n}\vert ({\cal E}_{-}(M)^q\is M)_s
\vert\right]<+\infty$. As a result, we deduce that $h{\cal E}_{-}(M)^q\is M^{T_n}$ is a uniformly integrable martingale. Thus, by combining this with  (\ref{HellingerEquality1}), we derive \begin{eqnarray*}
E\left[(h\is {\cal E}(M)^q)_{T\wedge T_n}\right]&&=q(q-1)E\left[\int_0^{T\wedge T_n}h_s {\cal E}_{s-}(M)^qd {\cal H}_s^{(q)}(M)\right],\end{eqnarray*}
Hence, a combination of this with $E\left[(h\is {\cal E}(M)^q)_T\right]=\sup_n E\left[(h\is {\cal E}(M)^q)_{T\wedge T_n}\right]$ --which always holds as soon as  ${\cal E}(M)^q$ is a local submartingale--  and the class monotone theorem, we deduce that (\ref{HellingerEquality3})  follows. This ends the proof of the lemma.\end{proof}
  \end{appendices}

\section*{Acknowledgements}
\noindent  This research is fully supported by the
Natural Sciences and Engineering Research Council of Canada, nGrant NSERC RGPIN-2019-04779. \\ 
\noindent The authors would like to thank  Ferdoos Alharbi, Safa' Alsheyab, Jun Deng, Youri Kabanov, Michele Vanmalele  and the participants of the Bachelier Coloquim 2020, for several comments/suggestions, fruitful discussions on the topic, and/or for providing important and useful related references.\\
\noindent The authors are very grateful to an anonymous referee for the careful reading, important suggestions, and pertinent comments that helped improving the paper. 





\end{document}